\documentclass[11pt,a4paper]{article}

\usepackage[mathscr]{euscript}\usepackage{times}
\usepackage{cite}
\usepackage{fullpage} 
\usepackage{amsmath}
\usepackage{amsfonts}
\usepackage{amssymb}
\usepackage{amsthm}
\usepackage{filecontents}  
\usepackage{algorithm}
\usepackage[noend]{algpseudocode}

\newcounter{thm}
\newtheorem{theorem}[thm]{Theorem}

\newtheorem{proposition}[thm]{Proposition}
\newtheorem{lemma}[thm]{Lemma}
\newtheorem{corollary}[thm]{Corollary}
\theoremstyle{definition}
\newtheorem{definition}[thm]{Definition}

\makeatletter
\newtheorem*{rep@theorem}{\rep@title}
\newcommand{\newreptheorem}[2]{
\newenvironment{rep#1}[1]{
 \def\rep@title{#2 \ref{##1}}
 \begin{rep@theorem}}
 {\end{rep@theorem}}}
 \newtheorem*{rep@corollary}{\rep@title}
\newcommand{\newrepcorollary}[2]{
\newenvironment{rep#1}[1]{
 \def\rep@title{#2 \ref{##1}}
 \begin{rep@corollary}}
 {\end{rep@corollary}}}
\makeatother

\newcommand\prob[3]{
  \begin{description}
    \item[\it Name.] #1
    \vspace*{-2mm}
    \item[\it Instance.] #2
    \vspace*{-2mm}
    \item[\it Output.] #3
  \end{description}
}

\newreptheorem{theorem}{Theorem}
\newreptheorem{lemma}{Lemma}
\newreptheorem{corollary}{Corollary}

\newcommand{\CSP}{\ensuremath{\mathrm{CSP}}}
\newcommand\nCSP{\ensuremath{\mathrm{\#CSP}}}
\newcommand\FP{\ensuremath{\mathrm{FP}}}

\newcommand\NP{\ensuremath{\mathrm{NP}}}

\newcommand\numP{\ensuremath{\mathrm{\#P}}}
\def\nP{\numP}

\newcommand\defn[1]{\emph{#1}}
\newcommand\LMPartitions[2]{{\sc \#$#1$-$#2$-partitions}}
\newcommand\nListPartitions[1]{{\sc \#List-$#1$-partitions}}
\newcommand\nPartitions[1]{{\sc \#$#1$-partitions}}
\newcommand{\CMPartitions}[2]{{\sc \#$#1$-$#2$-partitions}}
\newcommand\ListPartitions[1]{{\sc List-$#1$-partitions}}
\newcommand\LMPurify[2]{{\sc \#$#1$-$#2$-purify}}
\newcommand\LMPurifyStep[2]{{\sc \#$#1$-$#2$-purify-step}}
\newcommand\ExistsDerectSeq{{\sc ExistsDerectSeq}}
\newcommand\MatrixHasDerectSeq{{\sc MatrixHasDerectSeq}}

\newcommand\powerset[1]{{\mathcal{P}(#1)}}
\newcommand{\subclo}[1]{{\mathscr{S}(#1)}}
\newcommand\lists{\ensuremath{\mathcal L}}
\newcommand{\const}[1]{\langle #1 \rangle}
\newcommand{\GammaLM}{\Gamma_{\!\lists, M}}
\newcommand{\Gammabar}{\overline{\Gamma}}
\newcommand{\nats}{\mathbb{Z}_{\geq 0}}

\newcommand{\Mstart}{M_\mathrm{start}}
\newcommand{\Mend}{M_\mathrm{end}}
\newcommand{\Mbij}{M_\mathrm{bij}}
\newcommand{\Mzero}{\mathbf{0}}
\newcommand{\Mid}{\mathrm{Id}}
\newcommand{\Mhp}{\ensuremath{M_{\mathrm{hp}}}}  
\newcommand{\Mhs}{\ensuremath{M_{\mathrm{hs}}}}  
\def\newX{X}

\begin{document}

\title{Counting list matrix partitions of graphs\thanks{A preliminary version of this 
paper  appeared in the proceedings of CCC 2014.
The research leading to these results has
received funding from the 
MEXT Grants-in-Aid for
  Scientific Research and the EPSRC and  
the European Research Council under the European Union's Seventh Framework Programme (FP7/2007--2013) ERC grant agreement no.\ 334828. The paper 
reflects only the authors' views and not   the views of the ERC or the European Commission. The European Union is not liable for any use that may be made of the information contained therein. 
  }}
\author{Andreas G\"obel\thanks{Department of Computer Science,
University of Oxford, Wolfson Building, Parks Road, Oxford, OX1 3QD,
United Kingdom.} \and
Leslie Ann Goldberg\footnotemark[2] \and
Colin McQuillan\thanks{Department of Computer Science,
Ashton Building, University of Liverpool, Liverpool, L69 3BX,
United Kingdom.}  \and
David Richerby\footnotemark[2] \and
Tomoyuki Yamakami\thanks{
Department of Information Science, University of Fukui,
3-9-1 Bunkyo, Fukui City,
Fukui 910-8507, Japan.
 }}

\date{}

\maketitle
\begin{abstract} 
Given a symmetric $D\times D$ matrix $M$ over $\{0,1,*\}$, 
a list $M$-partition of a graph~$G$   
  is a partition of the vertices of~$G$ into~$D$ parts
which are associated with the rows of~$M$. 
The part of each vertex is chosen from
a given list in such a way that no edge of~$G$ is mapped to a $0$ in~$M$
and no non-edge of~$G$ is mapped to a $1$ in~$M$.
Many important graph-theoretic structures can be represented
as list $M$-partitions including graph colourings, split graphs 
and homogeneous sets and pairs, which arise in the proofs of the weak and strong perfect
graph conjectures. Thus, there has been quite a bit of work
on determining for which matrices $M$ computations involving list $M$-partitions are tractable.
This paper focuses on the problem of counting   list $M$-partitions, given a graph~$G$
and given 
a list
for each vertex of~$G$. 
We identify a certain set of ``tractable'' matrices~$M$.
We give an algorithm that counts list $M$-partitions in polynomial time
 for every (fixed)  matrix~$M$ in this set.
The algorithm relies on data structures such as sparse-dense partitions and
subcube decompositions to reduce each problem instance to a sequence
of problem instances in which the lists have a certain  useful structure
that restricts access to portions of~$M$ in which the interactions of $0$s and $1$s is controlled.
We show how to solve the resulting  restricted instances by converting
them into particular counting constraint satisfaction problems ($\nCSP$s)
which we show how to solve using a constraint satisfaction technique
known as ``arc-consistency''.
For every matrix~$M$ for which our algorithm fails,
we show that the problem of counting list $M$-partitions   is \numP{}-complete.
Furthermore, we give an explicit characterisation of the dichotomy theorem ---
counting list $M$-partitions is tractable (in \FP{}) if   the
matrix~$M$ has a structure called a derectangularising sequence.
If $M$ has no derectangularising sequence, we show that counting list $M$-partitions 
is \numP{}-hard.
We  show that
  the meta-problem of determining whether a given matrix has a
  derectangularising sequence is  \NP-complete.
Finally, we show that list $M$-partitions can be used to encode
cardinality restrictions in $M$-partitions problems and we use this to
give a polynomial-time algorithm for counting homogeneous pairs in graphs.
\end{abstract}

\section{Introduction}

A matrix partition of an undirected graph is a partition of its
vertices according to a matrix  which specifies adjacency and non-adjacency
conditions on
the vertices, depending on the parts to which they are assigned. 
For finite sets~$D$ and~$D'$,
the set $\{0,1,*\}^{D\times D'}$
is the set of matrices with rows 
indexed by~$D$
and columns indexed by~$D'$
where each $M_{i,j} \in \{0,1,*\}$.
For any symmetric matrix $M\in\{0,1,*\}^{D\times D}$, an
\defn{$M$-partition} of an undirected graph $G=(V,E)$ is a function $\sigma\colon V\to D$ such
that, for distinct vertices $u$ and~$v$,
\begin{itemize}
\item $M_{\sigma(u),\sigma(v)}\neq 0$ if $(u,v)\in E$ and
\item $M_{\sigma(u),\sigma(v)}\neq 1$ if $(u,v)\not\in E$.
\end{itemize}
Thus, $M_{i,j}=0$ means that no edges are allowed between 
vertices in
parts $i$ and~$j$, $M_{i,j}=1$ means 
that there must be an edge
between every pair of vertices in the two parts and $M_{i,j}=*$
means that any set of edges is allowed between the parts.  
For entries $M_{i,i}$ on the
diagonal of~$M$, the conditions only apply to distinct vertices in
part~$i$.   Thus, $M_{i,i}=1$ requires that the vertices in part~$i$ form a
clique in~$G$ and $M_{i,i}=0$ requires that they form an independent set.

For example, if $D=\{i,c\}$, $M_{i,i} = 0$, $M_{c,c}=1$ and $M_{c,i} =
M_{i,c} = *$, i.e., $M=\left(\begin{smallmatrix}0 & *\\ * &
1\end{smallmatrix}\right)$, then an $M$-partition of a graph is a
partition of its vertices into an independent set (whose vertices are mapped
to~$i$) and a clique (whose vertices are mapped to~$c$). The independent
set and the clique may have
arbitrary edges between them.  A graph that has such an $M$-partition is
known as a split graph~\cite{Golumbic}.

As Feder, Hell, Klein and Motwani  describe~\cite{FHKM},
many important graph-theoretic structures can be represented as 
$M$-partitions, including graph colourings, split graphs,
$(a,b)$-graphs~\cite{Bra96}, clique-cross partitions~\cite{EKR},
and their generalisations. $M$-partitions also arise as ``type partitions'' in extremal graph theory~\cite{BT00}.
In the special case where $M$ is a $\{0,*\}$-matrix (that is, it has no 1~entries), 
$M$-partitions of~$G$ correspond to homomorphisms 
from~$G$ to the (potentially looped) graph~$H$ whose  
adjacency matrix is obtained from $M$ by turning every $*$ into a~1.
Thus, proper $|D|$-colourings of~$G$ are exactly $M$-partitions
for the matrix~$M$ which has 0s on the diagonal and $*$s elsewhere.

To represent more complicated graph-theoretic structures, such as homogeneous sets
and their generalisations, which arise in the proofs of the weak and 
strong perfect graph conjectures~\cite{lovasz,CRST}, 
it is necessary to generalise $M$-partitions by introducing lists.
Details of these applications are given by Feder et~al.~\cite{FHKM}, who define the notion of a 
list $M$-partition.

A \defn{list $M$-partition} is an $M$-partition~$\sigma$ that is also required to
satisfy constraints on the values of each $\sigma(v)$.  
Let $\powerset{D}$ denote the powerset of~$D$.
We say that $\sigma$
\defn{respects} a function $L\colon V(G)\to \powerset{D}$ if $\sigma(v)\in
L(v)$ for all $v\in V(G)$.  Thus, for each vertex $v$, $L(v)$ serves
as a list of allowable parts for~$v$ and a \emph{list $M$-partition}
of~$G$ is an $M$-partition that respects the given list function.
We allow empty lists for technical convenience, although
there are no $M$-partitions that respect any list function~$L$ where $L(v)=\emptyset$
for some vertex~$v$.

Feder et~al.\cite{FHKM} study the computational complexity of the following
decision problem, which is parameterised by
a symmetric matrix $M\in\{0,1,*\}^{D\times D}\!$.
\prob{\ListPartitions{M}.}
 {A pair $(G,L)$ in which $G$ is a
  graph and $L$ is a function $V(G)\to\powerset{D}$.}  
  {``Yes'', if $G$ has an $M$-partition that respects $L$;
  ``no'', otherwise.}
Note that $M$ is a parameter of the problem rather than an input of the problem.
Thus, its size is a constant which does not vary with the input.

A series of papers~\cite{FH,FHHList,FHH} described in~\cite{FHKM}
presents a complete dichotomy for the
special case of homomorphism problems, which are  \ListPartitions{M} 
problems in which $M$ is a $\{0,*\}$-matrix.
In particular, Feder, Hell and Huang~\cite{FHH} show that, for every $\{0,*\}$-matrix~$M$
(and symmetrically, for every $\{1,*\}$-matrix~$M$), 
the problem \ListPartitions{M} is either polynomial-time solvable or
\NP{}-complete.

It is important to note that both of these special cases of \ListPartitions{M}
are constraint satisfaction problems (CSPs) and a famous conjecture of
Feder and Vardi~\cite{FV} is that a P versus \NP{}-complete dichotomy
also exists for every CSP.  
Although general \ListPartitions{M}
problems can also be coded as CSPs with restrictions on the input,\footnote{
For the reader who is familiar with CSPs, it might be useful to see how a \ListPartitions{M}
problem can be coded as a CSP with restrictions on the input.
Given a symmetric $M \in \{0,1,*\}^{D\times D}$,
let $M_0$ be the relation on $D\times D$ containing all pairs $(i,j) \in D\times D$
for which $M_{i,j} \neq 1$.
Let $M_1$ be the relation on $D\times D$ containing all pairs $(i,j)\in D\times D$
for which $M_{i,j} \neq 0$.
Then a \ListPartitions{M} problem with input $G,L$
can be encoded as a CSP whose constraint language includes the binary relations~$M_0$
and~$M_1$ and also the unary relations corresponding to the sets in the image of~$L$.
Each vertex $v$ of~$G$ is a variable in the CSP instance with the unary constraint $L(v)$.
If $(u,v)$ is an edge of~$G$ then it is constrained by~$M_1$. If it is a non-edge of~$G$,
it is constrained by~$M_0$.
Note that the CSP instance satisfies the restriction that every pair of distinct variables 
has exactly one constraint, which is either $M_0$ or $M_1$.
In a general CSP instance, a pair of variables could be constrained 
by $M_0$ and $M_1$ or one of them, or neither.
It is not clear how to code such a general CSP instance as a list partitions problem.} 
it
is not known how to code them without such restrictions.  Since the
Feder--Vardi conjecture applies only to CSPs with unrestricted inputs,
even if proved, it would not necessarily apply to  \ListPartitions{M}.

Given the many applications of  \ListPartitions{M}, it is important to know whether there is a dichotomy for this problem.
This is part of a major ongoing research effort which has the goal 
of understanding the boundaries of tractability by identifying classes of problems,
as wide as possible, where dichotomy theorems arise and where
the precise boundary between tractability and intractability can be specified.
  
Significant progress has been made on identifying dichotomies for
\ListPartitions{M}.
Feder et~al.~\cite[Theorem~6.1]{FHKM} give a 
complete dichotomy for the special case in which $M$ is at most $3\times 3$,
by showing that \ListPartitions{M} is polynomial-time solvable or \NP{}-complete
for each such matrix.   
Later, Feder and Hell studied the \ListPartitions{M} problem under the name
CSP$^*_{1,2}(H)$ and showed \cite[Corollary 3.4]{FHFull} that, for every~$M$,
\ListPartitions{M} is either \NP{}-complete, or is solvable in quasi-polynomial time.
 In the latter case, they showed that
\ListPartitions{M} is solvable in $n^{O(\log n)}$ time, given an 
$n$-vertex graph. Feder and Hell refer to this result as a ``quasi-dichotomy''.

Although the Feder--Vardi conjecture remains open,
a complete dichotomy is now known for counting CSPs.
In particular, Bulatov~\cite{Bul08} (see also \cite{DRfull})   
has shown that, for every constraint language~$\Gamma$,
the counting constraint satisfaction problem $\nCSP(\Gamma)$ is either polynomial-time
solvable, or \numP{}-complete.  
It is natural to ask whether a similar situation arises for counting list $M$-partition problems.
We study the following computational problem, which is parameterised by 
a finite symmetric matrix
$M\in\{0,1,*\}^{D\times D}\!$.

\prob{\nListPartitions{M}.}
 {A pair $(G,L)$ in which $G$ is a
  graph and $L$ is a function $V(G)\to\powerset{D}$.}  
  {The number of
  $M$-partitions of $G$ that respect $L$.}

Hell, Hermann and Nevisi~\cite{HHN} have considered the related problem \nPartitions{M}
without lists, which can be seen as
\nListPartitions{M} restricted to the case that $L(v)=D$ for every vertex~$v$.
This problem is defined as follows.

\prob{\nPartitions{M}.}
 {A  graph $G$.}  
  {The number of
  $M$-partitions of $G$.}
  
  In the problems \ListPartitions{M}, \nListPartitions{M} and \nPartitions{M},
 the matrix $M$ is fixed and its size does not vary with the input.

 Hell et~al.\ gave a   dichotomy for small matrices~$M$ (of size at most $3\times 3$).
 In particular, \cite[Theorem 10]{HHN} together with the graph-homomorphism dichotomy of
 Dyer and Greenhill~\cite{DG} shows that, for every such~$M$, \nPartitions{M} is either polynomial-time
 solvable or $\numP$-complete.
An interesting feature of counting $M$-partitions, identified by Hell et~al.\@
is that, unlike the situation for homomorphism-counting problems, there are tractable $M$-partition problems with non-trivial counting algorithms.
Indeed the main contribution of the present paper, as described below, is
to identify a set of ``tractable'' matrices~$M$
and to give
a non-trivial   algorithm which solves   \nListPartitions{M}
for every such~$M$.
We combine this with a proof
that \nListPartitions{M} is $\numP$-complete for every other~$M$.

\subsection{Dichotomy theorems for counting list $M$-partitions}
\label{subsec:dichotomy}

Our main theorem is  a general dichotomy for the counting list $M$-partition problem,
for matrices~$M$ of all sizes.  As noted above, since there is no
known coding of list $M$-partition problems as CSPs without input
restrictions, our theorem is not known to be implied by the
dichotomy for \nCSP{}.

Recall that \FP{} is the class of functions computed by
polynomial-time deterministic Turing machines.  
\numP{} is the
class of functions~$f$ for which there is a nondeterministic
polynomial-time Turing machine that has exactly $f(X)$ accepting paths for
every input~$X$; this class can be thought of as the natural analogue of \NP{} for counting problems.
Our main theorem is the following.

\newcommand\dichotomy{
For any symmetric matrix
  $M\in\{0,1,*\}^{D\times D}\!$, \nListPartitions{M} is
  either in $\FP$ or $\numP$-complete.}
\begin{theorem}\label{thm:dichotomy}\dichotomy\end{theorem}

To prove Theorem~\ref{thm:dichotomy}, we investigate the complexity of the
more general counting problem \LMPartitions{\lists}{M},
which has two parameters --- a  matrix $M\in\{0,1,*\}^{D\times D}$
and a (not necessarily proper) subset~\lists{} of~$\powerset{D}$.   
In this problem, we only allow sets in~\lists{} to be used as lists.

\prob{\LMPartitions{\lists}{M}.} 
{A pair $(G,L)$ where $G$ is a
  graph and $L$ is a function $V(G)\to \lists$.}  
  {The number of
  $M$-partitions of $G$ that respect $L$.}

Note that $M$ and~\lists{} are fixed parameters  
of \LMPartitions{\lists}{M} --- they
are not part of
the input instance.
The problem  \nListPartitions{M} is just the 
special
case of \LMPartitions{\lists}{M}
where $\lists = \powerset{D}$.

We say that a set $\lists \subseteq \powerset{D}$
is \defn{subset-closed} if $A\in \lists$
implies that every subset of $A$ is in \lists{}.
This closure property is referred to as the ``inclusive'' case in~\cite{FHFull}.

\begin{definition}\label{def:closure}
Given a set $\lists\subseteq \powerset{D}$, we write $\subclo{\lists}$
for its subset-closure,
which is the set 
$$\subclo{\lists}=\{X \mid \mbox{for some $Y\in \lists$,  $X\subseteq Y$}\}.$$
\end{definition}

We prove the following theorem,
which immediately implies Theorem~\ref{thm:dichotomy}.
  
  \newcommand\fulldichotomy
  {Let $M$ be a symmetric matrix in 
  $\{0,1,*\}^{D\times D}$
  and let $\lists\subseteq\powerset{D}$ be
  subset-closed.   The problem
  \LMPartitions{\lists}{M} is
  either in $\FP$ or $\numP$-complete.}
\begin{theorem}\label{thm:fulldichotomy}\fulldichotomy\end{theorem}

Note that this does not imply a dichotomy for the counting $M$-partitions
problem without lists.  The problem with no lists corresponds to the case where every
vertex of the input graph~$G$ is assigned the list~$D$, allowing the
vertex to be potentially placed in any part.  Thus, the problem
without lists is equivalent to 
the problem \LMPartitions{\lists}{M} with
$\lists=\{D\}$, but Theorem~\ref{thm:fulldichotomy} applies only to
the case where \lists{}~is subset-closed.

\subsection{Polynomial-time algorithms and an explicit dichotomy}

We now introduce the concepts needed to  give 
an explicit
criterion for the dichotomy in Theorem~\ref{thm:fulldichotomy}
and to provide polynomial-time algorithms for all tractable cases.
We use standard definitions of relations and their arities,
compositions and inverses.

\begin{definition} For any symmetric $M\in\{0,1,*\}^{D\times D}$ and any 
sets $X,Y\in\powerset{D}$, define the binary relation
\begin{equation*}
    H^M_{X,Y}=\{(i,j)\in X\times Y\mid M_{i,j}=*\}.
\end{equation*}
\end{definition}

The intractability condition for 
the problem \LMPartitions{\lists}{M} begins with the following notion of rectangularity, 
which was 
introduced by Bulatov and Dalmau~\cite{BD}.

\begin{definition}
A relation $R\subseteq D\times D'$ is \defn{rectangular} if, for all $i,j\in D$, and $i'\!,j'\in D'\!$,
$$(i,i'),(i,j'),(j,i')\in R\implies (j,j')\in R\,.$$
\end{definition}

Note that the intersection of two rectangular relations is itself
rectangular.  However, the composition of two rectangular relations
is not necessarily rectangular: for example, $\{(1,1), (1,2),
(3,3)\}\circ \{(1,1), (2,3), (3,1)\} = \{(1,1), (1,3), (3,1)\}$.

Our dichotomy criterion will be based on what we call
\lists{}-$M$-derectangularising sequences.
In order to define these, we introduce the notions of pure matrices and $M$-purifying sets.
\begin{definition}
Given index sets $X$ and~$Y$, a matrix $M\in\{0,1,*\}^{X\times Y}$ is
\defn{pure} if it has no $0$s or has no $1$s.
\end{definition}

Pure matrices correspond to ordinary graph homomorphism problems.   
As we noted above, 
$M$-partitions of~$G$ correspond to homomorphisms of~$G$ when
$G$ is
a $\{0,*\}$-matrix. The same is true (by complementation)
when
$G$ is
a $\{1,*\}$-matrix.

\begin{definition}
For any
$M\in\{0,1,*\}^{D\times D}$, a 
set $\lists \subseteq \powerset{D}$ 
is \defn{$M$-purifying} if, for all $X,Y\in\lists$, the $X$-by-$Y$
submatrix $M|_{X\times Y}$ is pure.
\end{definition}

For example, consider the matrix
\begin{equation*}
    M = \left(\begin{matrix}
                  1 & * & 0 \\
                  * & 1 & * \\
                  0 & * & 1
              \end{matrix}\right)
\end{equation*}
with rows and columns indexed by $\{0,1,2\}$ in the obvious way.  The
matrix~$M$ is not pure but for
 $\lists =  \{\{0,1\},
\{2\}\}$,
the set $\lists$ is $M$-purifying
and so is the closure~$\subclo{\lists}$.

\begin{definition}
\label{def:derect}
An \defn{\lists{}-$M$-derectangularising sequence} of length $k$ is a
sequence $D_1,\dots,D_k$  with each $D_i \in\lists$ such that:
\begin{itemize}
\item  $\{D_1,\ldots,D_k\}$ is $M$-purifying and
\item the relation $H^M_{D_1,D_2} \circ H^M_{D_2, D_3} \circ \dots \circ
H^M_{D_{k-1}, D_k}$ is not rectangular.
\end{itemize}
\end{definition}
 
If there is an $i\in \{1,\ldots,k\}$ such that 
$D_i$ is the empty set 
then the relation $H=H^M_{D_1,D_2} \circ H^M_{D_2, D_3} \circ \dots \circ
H^M_{D_{k-1}, D_k}$
is the empty relation, which is trivially
rectangular. 
If there is an~$i$ such that $|D_i|=1$ then
$H$ is a Cartesian product, and is therefore rectangular.
It follows   that $|D_i|\geq 2$ for each~$i$ in a derectangularising sequence.

We can now state our explicit dichotomy theorem, which implies Theorem~\ref{thm:fulldichotomy} and,
hence, Theorem~\ref{thm:dichotomy}.

\newcommand\explicitdichotomy{
Let $M$ be a symmetric matrix in 
  $\{0,1,*\}^{D\times D}$
  and let $\lists{}\subseteq\powerset{D}$ be
  subset-closed.
If there is an
\lists{}-$M$-derectangularising sequence then
the problem \LMPartitions{\lists}{M} is $\numP$-complete.  Otherwise,
 it is in $\FP$.  
}
\begin{theorem}\label{thm:explicitdichotomy}\explicitdichotomy\end{theorem}

Sections~\ref{sec:purifiedcsp}, \ref{sec:arc} and 
\ref{sec:dichotomy}  develop a polynomial-time
algorithm which solves the problem \LMPartitions{\lists}{M}
whenever there is no \lists{}-$M$-derectangularising sequence.
The algorithm involves several steps.

First, consider the case in which 
$\lists$ is subset-closed and $M$-purifying.
In this case, Proposition~\ref{prop:purifiediscsp}
presents a polynomial-time transformation
from an instance of the problem \LMPartitions{\lists}{M}
to an instance of a related counting CSP.
Algorithm~\ref{alg:AC} exploits
special properties of the constructed CSP instance
so that it can be solved in polynomial time
using a CSP technique called arc-consistency.
(This is proved in Lemma~\ref{lem:quickarc}.)
This provides a solution to the original  \LMPartitions{\lists}{M}
problem for the $M$-purifying case.

The case in which $\lists$ is not $M$-purifying is
tackled in Section~\ref{sec:dichotomy}.
Section~\ref{sec:DS} gives algorithms for constructing the relevant
data structures, which include a special case of sparse-dense partitions
and also subcube decompositions.
Algorithm~\ref{alg:purify}
uses these data structures (via Algorithms~\ref{alg:purifystep},
\ref{alg:Case1}, \ref{alg:Case2}, \ref{alg:Case3} and~\ref{alg:purifytriv}) to reduce
the \LMPartitions{\lists}{M} problem 
to a sequence of problems
\LMPartitions{\lists_i}{M} where $\lists_i$ is $M$-purifying. 
Finally, the 
polynomial-time algorithm is presented in Algorithms~\ref{alg:mainpurifying}
and~\ref{alg:main}. For every   
$\lists$ and~$M$ 
 where there is no $\lists$-$M$-derectangularising sequence, 
either Algorithm~\ref{alg:mainpurifying} or Algorithm~\ref{alg:main}
defines a polynomial-time
function \LMPartitions{\lists}{M} for
solving the \LMPartitions{\lists}{M} problem, 
given an input $(G,L)$.
The function \LMPartitions{\lists}{M} is not recursive.
However, its \emph{definition} is recursive in the sense that
the function \LMPartitions{\lists}{M} defined in Algorithm~\ref{alg:main}
calls
a function \LMPartitions{\lists_i}{M} where $\lists_i$ is
a subset of $\powerset{D}$ whose cardinality is smaller than $\lists$.
The function 
\LMPartitions{\lists_i}{M} is, in turn, defined either in Algorithm~\ref{alg:mainpurifying}
or in~\ref{alg:main}.

The proof of Theorem~\ref{thm:explicitdichotomy}
shows that, when Algorithms~\ref{alg:mainpurifying} and~\ref{alg:main}
fail to solve 
the problem \LMPartitions{\lists}{M}, the problem
is $\numP$-complete.

\subsection{Complexity of the dichotomy criterion}

Theorem~\ref{thm:explicitdichotomy} gives a precise criterion under which
the problem
\LMPartitions{\lists}{M} is in $\FP$ or $\numP$-complete, where \lists{} and~$M$ are considered to be fixed parameters.
In Section~\ref{sec:meta}, we address the 
computational problem of determining which is the case, now treating \lists{} and~$M$ as inputs to this ``meta-problem''.
Dyer and Richerby~\cite{DRfull} studied the corresponding problem for the \nCSP{} dichotomy, showing that determining whether a constraint language~$\Gamma$ satisfies
the criterion for their $\nCSP(\Gamma)$ dichotomy is reducible to the graph automorphism problem, which is in \NP{}.
We are interested in the following computational problem, which we show to be \NP{}-complete.

\prob{\ExistsDerectSeq.} 
{An index set $D$, a symmetric matrix $M$ in $\{0,1,*\}^{D\times D}$ (represented as an array)
and 
a set $\lists{}\subseteq\powerset{D}$ (represented as a list of lists).}
  {``Yes'', if there is an $\subclo{\lists}$-$M$-derectangularising sequence; ``no'', otherwise. }

\newcommand\metathm{
\ExistsDerectSeq\ is \NP{}-complete under polynomial-time many-one reductions.}
\begin{theorem}\label{thm:meta}\metathm\end{theorem}
 
Note that, in the definition of the problem \ExistsDerectSeq, 
the input $\lists$ is not necessarily  subset-closed.  
Subset-closedness
allows a concise representation of some
inputs: for example, $\powerset{D}$ has exponential size but it can be
represented as $\subclo{\{D\}}$, so the corresponding
input is just $\lists=\{D\}$.
In fact, our proof of Theorem~\ref{thm:meta} uses a set of lists \lists{} where
$|X|\leq 3$ for all $X\in\lists$.  Since there are at most $|D|^3+1$
such sets, our \NP{}-completeness proof would still hold if we
insisted that the input \lists\ to \ExistsDerectSeq{} must be subset-closed.

Let us return to the original problem \nListPartitions{M}, which is the special case of the problem \LMPartitions{\lists}{M} where $\lists=\powerset{D}$.  This leads us to be interested in the following computational problem.
\prob{\MatrixHasDerectSeq.} 
{An index set $D$ and a symmetric matrix $M$ in $\{0,1,*\}^{D\times D}$ (represented as an array).}
  {``Yes'', if there is a $\powerset{D}$-$M$-derectangularising sequence; ``no'', otherwise. }

Theorem~\ref{thm:meta} does not quantify the complexity
of \MatrixHasDerectSeq{} because its proof relies on a specific choice
of \lists{}
which, as we have noted, is not $\powerset{D}$.
Nevertheless, the proof of Theorem~\ref{thm:meta}
has the following corollary.
 
\newcommand\metacor{
\MatrixHasDerectSeq\ is in \NP{}.}
\begin{corollary}\label{cor:meta}\metacor\end{corollary}

\subsection{Cardinality constraints}

Many combinatorial structures can be represented as $M$-partitions
with the addition of cardinality constraints on the parts.  For
example, it might be required that certain parts be non-empty or, more
generally, that they contain at least~$k$ vertices for some fixed~$k$.

Feder et~al.~\cite{FHKM} showed that
the problem of determining whether such a structure exists in 
a given graph can be reduced to a  \ListPartitions{M} problem
in which the cardinality constraints are expressed using lists.  In
Section~\ref{sec:card}, we extend this to counting.  We show that any
\nPartitions{M} problem with additional cardinality constraints of the
form, ``part~$d$ must contain at least $k_d$ vertices'' is
polynomial-time Turing reducible to \nListPartitions{M}.  As a
corollary, we show that the ``homogeneous pairs'' introduced by
Chv\'atal and Sbihi~\cite{CS1987:Bull-free} can be counted in
polynomial time.  Homogeneous pairs can be expressed as an
$M$-partitions problem for a certain $6\times 6$~matrix, with
cardinality constraints on the parts.  
 
\section{Preliminaries}
\label{sec:prelim}

For a positive integer~$k$, we write~$[k]$ to denote the set
$\{1,\dots,k\}$.  
If $\mathcal{S}$ is a set of sets then we use
$\bigcap \mathcal{S}$ to denote the intersection of all sets in~$\mathcal{S}$.
The vertex set of a graph~$G$ is denoted $V(G)$ and
its edge set is
$E(G)$. We write $\{0,1,*\}^{D}$ for the set of all
functions $\sigma\colon D\to\{0,1,*\}$ and 
$\{0,1,*\}^{D\times D'}$ for the set of all matrices $M=(M_{i,j})_{i\in
  D,j\in D'}$, where each $M_{i,j}\in\{0,1,*\}$.

We always use the term ``$M$-partition'' when talking about a
partition of the vertices of a graph according to a
$\{0,1,*\}$-matrix~$M$.  When we use the term ``partition'' without
referring to a matrix, we mean it in the conventional sense of partitioning a
set~$X$ into disjoint subsets $X_1, \dots, X_k$ with $X_1\cup
\dots \cup X_k = X$.

We   view computational counting problems  as
functions mapping strings over input alphabets to natural numbers.  Our
model of computation is the standard multi-tape Turing machine. We say
that a counting problem~$P$ is polynomial-time Turing-reducible to
another counting problem~$Q$ if there is a polynomial-time
deterministic oracle Turing machine $M$ such that, on every instance
$x$ of $P$, $M$ outputs $P(x)$ by making queries to oracle~$Q$. We say
that $P$ is polynomial-time Turing-equivalent to~$Q$ if each is
polynomial-time Turing-reducible to the other.  For decision problems
(languages), we use the standard many-one reducibility: language $A$
is many-one reducible to language $B$ if there exists a function~$f$
that is computable in polynomial time such that $x\in A$ if and only if
$f(x)\in B$.

\section{Counting list $M$-partition problems and counting CSPs}
\label{sec:purifiedcsp}

Toward the 
development of our algorithms and the
proof of our dichotomy, we study a
special case of the problem \LMPartitions{\lists}{M}, in which \lists{} is
$M$-purifying and subset-closed. For such \lists{} and~$M$,
we show that 
the problem \LMPartitions{\lists}{M} is polynomial-time Turing-equivalent to a
counting constraint satisfaction problem (\nCSP{}).
To give the equivalence, we  
introduce the notation needed to specify \#CSP{}s.

A \emph{constraint language} is a finite set $\Gamma$ of named relations over some set~$D$.  For such a language, we define the
counting problem $\nCSP(\Gamma)$ as follows.

\prob{$\nCSP(\Gamma)$.} 
{A set $V$ of variables and a
  set $C$ of constraints of the form $\const{(v_1,\dots,v_k),R}$, where 
    $(v_1,\dots,v_k)\in V^k$ and $R$ is an
  arity-$k$ relation in $\Gamma$.}  
  {The number of assignments
  $\sigma\colon V\to D$ such that
\begin{equation}\label{eq:satisfying}(\sigma(v_1),\dots,\sigma(v_k))\in
  R\text{ for all }\const{(v_1,\dots,v_k),R}\in C\,.
\end{equation}}
The tuple of variables $v_1, \dots, v_k$ in a constraint is referred
to as the constraint's \emph{scope}.
The assignments $\sigma\colon V\to D$ for which
\eqref{eq:satisfying} holds are called the \defn{satisfying assignments} of 
the instance $(V,C)$.  Note that a unary constraint $\const{v,R}$ has
the same effect as a list: it directly restricts the possible values
of the variable~$v$.  As before, we allow the possibility that
$\emptyset\in\Gamma$; 
any instance that includes a constraint
$\const{(v_1, \dots, v_k), \emptyset}$ has no satisfying assignments.

\begin{definition} 
\label{defgammaprime}
Let $M$ be a symmetric matrix in $\{0,1,*\}^{D\times D}$
and let \lists{} be a
subset-closed $M$-purifying set.
Define the constraint language
$$\Gamma'_{\!\lists,M} = \{H^M_{X,Y}\mid X,Y\in\lists\}$$
and let $\GammaLM = \Gamma'_{\!\lists,M} \cup \powerset{D}$, where
$\powerset{D}$ represents the set of all unary relations on~$D$. 
\end{definition}

The unary constraints in $\GammaLM$ will be useful in our study of the
complexity of the dichotomy criterion, in Section~\ref{sec:meta}.
First, we define a convenient restriction on instances of
$\nCSP(\GammaLM)$.

\begin{definition}
\label{def:simple}
    An instance of $\nCSP(\GammaLM)$ is \emph{simple} if:
    \begin{itemize}
    \item there is exactly one unary constraint $\const{v,X_v}$ for each
        variable $v\in V\!$,
    \item there are no binary constraints $\const{(v,v),R}$, and
    \item each pair $u$,~$v$ of distinct variables appears in at most
        one constraint of the form $\const{(u,v),R}$ or $\const{(v,u),R}$.
    \end{itemize}
\end{definition}

\begin{lemma}
\label{lemma:simple}
    For every instance $(V,C)$ of $\nCSP(\GammaLM)$, there is a simple
    instance $(V,C')$ such that an assignment $\sigma\colon V\to D$
    satisfies $(V,C)$ if and only if it satisfies $(V,C')$.
    Further, such an instance can be computed in polynomial time.
\end{lemma}
\begin{proof}
    Observe that the set of binary relations in $\GammaLM$ is closed
    under intersections: $H^M_{X,Y} \cap H^M_{X'\!,Y'} =
    H^M_{X\cap X'\!,Y\cap Y'}$ and this relation is in~$\GammaLM$
    because \lists{} is subset-closed.  The binary part of $\GammaLM$
    is also closed under relational inverse because $M$ is symmetric,
    so
    \begin{equation*}
        \left(H^M_{X,Y}\right)^{-1} = \{(b,a) \mid (a,b)\in H^M_{X,Y}\}
                                              = H^M_{Y,X}\in \GammaLM\,.
    \end{equation*}
    Since $\powerset{D}\subseteq \GammaLM$, the set of unary relations
    is also closed under intersections.

    We construct $C'$ as follows, starting with~$C$.  Any binary
    constraint $\const{(v,v), R}$ can be replaced by the unary
    constraint $\const{v, \{d\mid (d,d)\in R\}}$.  All the binary
    constraints between distinct variables $u$ and~$v$ can be replaced
    by the single constraint
    \begin{equation*}
        \left\langle
            (u,v), \bigcap 
                \{R \mid \const{(u,v), R}\in C
                    \text{ or } \const{(v,u), R^{-1}}\in C\}
        \right\rangle\,.
    \end{equation*}
    Let the set of constraints produced so far be $C''\!$.  For each
    variable $v$ in turn, if there are no unary constraints applied
    to~$v$ in $C''\!$, add the constraint $\const{v, D}$; otherwise,
    replace all the unary constraints involving $v$ in~$C''$ with
    the single constraint
    \begin{equation*}
        \left\langle
            v, \bigcap \{R \mid \const{v, R}\in C''\}
        \right\rangle\,.
    \end{equation*}
    $C'$ is the resulting constraint set.  The closure properties
    established above guarantee that $(V,C')$ is a $\nCSP(\GammaLM)$
    instance.  It is clear that it has the same satisfying assignments
    as $(V,C)$ and that it can be produced in polynomial time.
\end{proof}

Our main result connecting the counting list $M$-partitions problem with
counting CSPs is the following.

\newcommand\purifiediscsp{For any symmetric $M\in\{0,1,*\}^{D\times
    D}$ and any subset-closed, $M$-purifying set \lists{}, the problem
  \LMPartitions{\lists}{M} is polynomial-time Turing-equivalent to
  $\nCSP(\GammaLM)$.}
\begin{proposition}\label{prop:purifiediscsp}\purifiediscsp\end{proposition}

Because of its length, we split the proof of the proposition into two
lemmas.

\begin{lemma}
  For any symmetric $M\in\{0,1,*\}^{D\times D}$ and any subset-closed,
  $M$-purifying set \lists{}, $\nCSP(\GammaLM)$
  is polynomial-time Turing-reducible to \LMPartitions{\lists}{M}.
\end{lemma}
\begin{proof}
Consider an input $(V,C)$ to $\nCSP(\GammaLM)$, which we may assume to
be simple.  Each variable appears in exactly one unary constraint,
$\const{v,X_v}\in C$.  Any variable $v$ that is not used in a binary
constraint can take any value in~$X_v$ so just introduces a
multiplicative factor of $|X_v|$ to the output of the counting CSP.
Thus, we will assume without loss of generality that every variable is
used in at least one constraint with a relation from
$\Gamma'_{\!\lists,M}$ and, by simplicity, there are no constraints of
the form $\const{(v,v),R}$.

We now define a corresponding instance $(G,L)$ of 
the problem \LMPartitions{\lists}{M}.
The vertices of $G$ are the variables~$V$ of the \nCSP{} instance.
For each variable $v\in V\!$, set
\begin{equation*}
    L(v) = X_v \cap \bigcap\left\{
               X \mid 
               \mbox{for some $u$ and $Y$,
              $ \const{(v,u),H^M_{X,Y}}\in C$  or
                     $ \const{(u,v),H^M_{Y,X}}\in C$}
           \right\}.
\end{equation*}
The edges $E(G)$ of our instance are the 
unordered pairs $\{u,v\}$ that satisfy one of the following conditions:
\begin{itemize}
\item there is a constraint between $u$ and $v$ in $C$ and $M|_{L(u)\times L(v)}$
  has a $0$ entry, or
\item there is no constraint between $u$ and $v$ in $C$ and $M|_{L(u)\times L(v)}$
has a $1$ entry.
\end{itemize}

Since every vertex~$v$ is used in at least one constraint with a relation
$H^M_{X,Y}$ where, by definition, $X$ and $Y$ are in $\lists$, every 
set~$L(v)$
is a subset of some 
set~$W\in\lists$.  
\lists{} is subset-closed so $L(v)\in\lists$
for all $v\in V$, as required.

We claim that a function
$\sigma\colon V\to D$ is a satisfying assignment of $(V,C)$ if and
only if it is an $M$-partition of $G$ that respects $L$.  Note that,
since \lists{} is $M$-purifying, no submatrix $M|_{X\times Y}$
($X,Y\in\lists)$ contains both 0s and~1s.

First, suppose that $\sigma$ is a satisfying assignment of $(V,C)$.
For each variable~$v$, $\sigma$ satisfies all the constraints
$\const{v,X_v}$, $\const{(v,u),H^M_{X,Y}}$ and
$\const{(u,v),H^M_{Y,X}}$ containing~$v$.  Therefore,
$\sigma(v)\in X_v$ and $\sigma(v)\in X$ for each binary constraint
$\const{(v,u),H^M_{X,Y}}$ or $\const{(u,v),H^M_{Y,X}}$, so
$\sigma$~satisfies all the list requirements.

To show that $\sigma$ is an $M$-partition of $G$,
consider any pair of distinct vertices $u,v\in V$.  If
there is a constraint $\const{(u,v), H^M_{X,Y}}\in C$, then
$\sigma$~satisfies this constraint so $M_{\sigma(u),\sigma(v)}=*$ and
$u$ and~$v$ cannot stop $\sigma$~being an $M$-partition.  Conversely,
suppose there is no constraint between~$u$ and~$v$ in~$C$.  If
$M|_{L(u)\times L(v)}$~contains a~0, there is no edge $(u,v)\in E(G)$ by
construction; otherwise, if $M|_{L(u)\times L(v)}$ contains a~1,
there is an edge $(u,v)\in E(G)$ by construction; otherwise, $M_{x,y}=*$ for
all $x\in L(u)$, $y\in L(v)$.  In all three cases, the assignment to
$u$ and~$v$ is consistent with $\sigma$~being an $M$-partition.

Conversely, suppose that~$\sigma$ is not a satisfying assignment
of~$(V,C)$.  If $\sigma$~does not satisfy some unary constraint
$\const{v,X}$ then $\sigma(v)\notin L(v)$ so $\sigma$ does not
respect~\lists{}.  If $\sigma$ does not satisfy some binary constraint
$\const{(u,v), H^M_{X,Y}}$ where $u$ and~$v$ are distinct then, by
definition of the relation $H^M_{X,Y}$, $M_{\sigma(u),\sigma(v)}\neq
*$.  If $M_{\sigma(u),\sigma(v)}=0$, there is an edge $(u,v)\in E(G)$ by
construction, which is forbidden in $M$-partitions; if
$M_{\sigma(u),\sigma(v)}=1$, there is no edge $(u,v)\in E(G)$
but this edge is required in $M$-partitions.  Hence,
$\sigma$~is not an $M$-partition.
\end{proof}
 
\begin{lemma}
  For any symmetric $M\in\{0,1,*\}^{D\times D}$ and any subset-closed,
  $M$-purifying set \lists{}, the problem \LMPartitions{\lists}{M} is
  polynomial-time Turing-reducible to $\nCSP(\GammaLM)$.
\end{lemma}
\begin{proof}
We now essentially reverse the construction of the previous lemma to
give a reduction from \LMPartitions{\lists}{M} to $\nCSP(\GammaLM)$.
For any instance ($G,L)$ of
\LMPartitions{\lists}{M},
we construct a corresponding instance $(V,C)$ of $\nCSP(\GammaLM)$ as follows.
The set of variables~$V$ is $V(G)$.
The set of constraints~$C$ consists of a constraint $\const{v,
L(v)}$ for each vertex $v\in V(G)$ and a constraint $\const{(u,v),
H^M_{L(u),L(v)}}$ for every pair of distinct vertices $u$, $v$ such that:
\begin{itemize}
\item $(u,v)\in E(G)$  and $M|_{L(u)\times L(v)}$
  has a 0 entry, or
\item  $(u,v)\not\in E(G)$ and $M|_{L(u)\times L(v)}$  has a 1 entry.
\end{itemize}

We show that a function
$\sigma\colon V\to D$ is a satisfying assignment of $(V,C)$ if and
only if it is an $M$-partition of $G$ that respects $L$.  It is clear
that $\sigma$~satisfies the unary constraints if and only if it
respects~$L$.

If $\sigma$ satisfies $(V,C)$ then consider any pair of distinct
vertices $u,v\in V$.  If there is a binary constraint involving $u$
and~$v$, then $M_{\sigma(u),\sigma(v)} = M_{\sigma(v),\sigma(u)} = *$
so the existence or non-existence of the edge $(u,v)$ of~$G$ does not
affect whether $\sigma$~is an $M$-partition.  If there is no binary
constraint involving $u$ and~$v$, then either there is an edge $(u,v)\in
E(G)$ and $M_{\sigma(u),\sigma(v)}\neq 0$ or there is no edge~$(u,v)$ and
$M_{\sigma(u),\sigma(v)}\neq 1$.  In all three cases, $\sigma$~maps $u$
and~$v$ consistently with it being an $M$-partition.

Conversely, if $\sigma$ does not satisfy $(V,C)$, either it fails to
satisfy a unary constraint, in which case it does not respect~$L$, or
it satisfies all unary constraints (so it respects~$L$), but it fails to satisfy a binary constraint $\const{(u,v),H^M_{L(u),L(v)}}$.  In the
latter case, by construction, $M_{\sigma(u),\sigma(v)}\neq *$ so
either $M_{\sigma(u),\sigma(v)}=0$ but there is an edge $(u,v)\in E(G)$, or
$M_{\sigma(u),\sigma(v)}=1$ and there is no edge $(u,v)\in E(G)$.  In either
case, $\sigma$~is not an $M$-partition of~$G$.
\end{proof}

\section{An arc-consistency based algorithm for $\nCSP(\GammaLM)$}
\label{sec:arc}

In the previous section, we showed that a class of
\LMPartitions{\lists}{M} problems is equivalent to a certain class of
counting CSPs, where the constraint language consists of binary
relations and all unary relations over the domain~$D$.  We now
investigate the complexity of such \nCSP{}s.

Arc-consistency is a standard solution technique for constraint
satisfaction problems~\cite{CSPbook}.  It is, essentially, a local
search method which initially assumes that each variable may take any
value in the domain and iteratively reduces the range of values that
can be assigned to each variable, based on the constraints applied to
it and the values that can be taken by other variables in the scopes
of those constraints.

For any simple $\nCSP(\GammaLM)$ instance~$(V,C)$,
define the vector of \defn{arc-consistent domains} $(D_v)_{v\in V}$ by the 
procedure in Algorithm~\ref{alg:ACComp}.
\begin{algorithm}[t]
\caption{The algorithm for computing arc-consistent domains 
for a simple $\nCSP(\GammaLM)$ instance~$(V,C)$ where, for each $v\in V$, $\const{v,X_v}\in C$ is the unary constraint involving~$v$.}
\label{alg:ACComp}
\begin{algorithmic}
\For{$v\in V$} 
\State{$D_v \gets X_v$}
\EndFor
\Repeat
 \For{$v\in V$}
   \State{$D'_v \gets D_v$}
 \EndFor
 \For{$\const{(u,v), R}\in C$} 
   \State{$D_u \gets \{d\in D_u \mid 
   \mbox{for some $d'\in D_v$, $(d,d')\in R$}\}$}
   \State{$D_v \gets \{d\in D_v \mid 
   \mbox{for some $d'\in D_u$, $(d',d)\in R$}\}$}
 \EndFor
\Until {$\forall {v\in V}$, $D_v=D'_v$}
\State{\Return {$(D_v)_{v\in V}$}}
\end{algorithmic}
\end{algorithm}
 At no point in the execution of the algorithm can any domain~$D_v$ increase in
size so, for fixed~$D$, the running time of the algorithm is at most a polynomial in 
$|V|+|C|$.

It is clear that, if $(D_v)_{v\in V}$ is the vector of arc-consistent
domains for a simple $\nCSP(\GammaLM)$   instance $(V,C)$, then every satisfying
assignment~$\sigma$ for that instance must have $\sigma(v)\in D_v$ for
each variable~$v$.  In particular, if some $D_v=\emptyset$, then the
instance is unsatisfiable.  (Note, though, that the converse does not
hold.  If $D=\{0,1\}$ and $R=\{(0,1),(1,0)\}$, the instance with constraints
$\const{x,D}$, $\const{y,D}$, $\const{z,D}$,
$\const{(x,y),R}$, $\const{(y,z),R}$ and $\const{(z,x),R}$ is
unsatisfiable but arc-consistency assigns $D_x = D_y = D_z =
\{0,1\}$.)

The arc-consistent domains computed for a simple instance $(V,C)$ can
yield further simplification of the constraint structure, which we
refer to as \defn{factoring}.  
The factoring applies when the arc-consistent domains restrict a binary
relation to a Cartesian product. In this case, the binary relation
can be replaced with corresponding unary relations.
Algorithm~\ref{alg:factor}  
factors a simple instance with
respect to a vector $(D_v)_{v\in V}$ of arc-consistent domains, producing a
set~$F$ of factored constraints.
\begin{algorithm}[t]
\caption{The algorithm for factoring
a simple $\nCSP(\GammaLM)$ instance~$(V,C)$
with respect to a vector $(D_v)_{v\in V}$ of arc-consistent domains.
$F$ is the set of factored constraints.}
\label{alg:factor}
\begin{algorithmic}
 \State{$F \gets C$}
 \For{$\const{(u,v), R} \in C$} 
   \If{$R\cap (D_u\times D_v)$ is a Cartesian product $D'_u\times
    D'_v$} 
        \State{Let $\const{u,\newX_u}$ and $\const{v,\newX_v}$ be the unary constraints 
        involving~$u$ and~$v$ in~$F$.}
         \State{$F \gets 
        (F \cup \{\const{u,\newX_u \cap D'_u}, \const{v,\newX_u \cap D'_v}\})
            \setminus \{ \const{(u,v), R}, \const{u,\newX_u}, \const{v,\newX_v}\}$}
     \EndIf
 \EndFor
 \State{\Return{$F$}}
\end{algorithmic}
\end{algorithm}
Recall that there is at most one constraint in~$C$ between distinct
variables and there are no binary constraints $\const{(v,v), R}$
because the instance is simple.  
Note also that, if $|D_u|\leq 1$ or
$|D_v|\leq 1$, then $R\cap (D_u\times D_v)$ is necessarily a Cartesian product.
It is easy to see that 
the result of factoring a simple instance is
simple, that Algorithm~\ref{alg:factor} runs in polynomial time and that  the instance~$(V,F)$ has the same satisfying assignments as~$(V,C)$.

The \emph{constraint graph} of a \CSP{} instance $(V,C)$ (in any
constraint language) is the undirected graph with vertex set~$V$
that contains an
edge between every pair of distinct variables that appear 
together in the scope
of some constraint.

\begin{algorithm}[t]
\caption{The arc-consistency based algorithm for counting satisfying assignments to
  simple instances of $\nCSP(\GammaLM)$.
  The input is a
simple
   instance $(V,C)$ of $\nCSP(\GammaLM)$.}
\label{alg:AC}
\begin{algorithmic}
    \Function{AC}{\text{variable set }V, \text{constraint set } C}
    \State{Use Algorithm~\ref{alg:ACComp} to compute the vector of arc-consistent domains $(D_v)_{v\in V}$}
    \State{Use Algorithm~\ref{alg:factor} to construct the set $F$ of factored constraints}
     \If{$D_v=\emptyset$ for some $v\in V$}
     \State{\Return{0}}
    \EndIf      
    \State{Compute the constraint graph $H$ of $(V,F)$}
        \State{Let  $H_1, \dots, H_\kappa$ be the components of $H$ with $V_i=V(H_i)$}
    \State{Let $F_i$ be the set of constraints in $F$ involving variables in $V_i$}
    \For{$i\in [\kappa]$}
        \If{$|D_w|=1$ for some $w\in V_i$}
            \State{$Z_i\gets 1$}
        \Else
            \State{Choose $w_i\in V_i$}\label{line:choice}
            \State{Let $\theta_i$ be the unary constraint involving $w_i$ in $F_i$}
            \For{$d\in D_{w_i}$}
               \State{$F'_{i,d} \gets (F_i \cup \{\const{w_i, \{d\}}\} ) \setminus \{\theta_i\}$}
            \EndFor            
            \State{$Z_i\gets \sum_{d\in D_{w_i}}
                   \mathrm{AC}(V_i,  F'_{i,d})$}
        \EndIf
    \EndFor
    \State{\Return{$\prod_{i=1}^\kappa Z_i$}}
    \EndFunction
\end{algorithmic}
\end{algorithm}

 Algorithm~\ref{alg:AC}  uses
arc-consistency to count the
satisfying assignments of simple $\nCSP(\GammaLM)$ instances.  It is
straightforward to see that the algorithm terminates, since each
recursive call is either on an instance with strictly fewer variables
or on one in which at least one variable has had its unary constraint
reduced to a singleton and no variable's unary constraint has
increased.  For general inputs, the algorithm may take exponential
time to run but, in Lemma~\ref{lem:quickarc} we show that the running
time is polynomial for the inputs we are interested in.

We first argue that the algorithm is correct.  By
Lemma~\ref{lemma:simple}, we may assume that the given instance
$(V,C)$ is simple.  Every satisfying assignment $\sigma\colon V\to D$
satisfies $\sigma(v)\in D_v$ for all $v\in V$ so restricting our
attention to arc-consistent domains does not alter the output.
Factoring the constraints also does not change the number of
satisfying assignments: it merely replaces some binary constraints
with equivalent unary ones.  The constraints are factored, so any
variable $v$ with $|D_v|=1$ must, in fact, be an isolated vertex in
the constraint graph because, as noted above, any binary constraint
involving it has been replaced by unary constraints.  Therefore, if a
component $H_i$ contains a variable $v$ with $|D_v|=1$, that component
is the single vertex~$v$, which is constrained to take a single value,
so the number of satisfying assignments for this component, which we denote~$Z_i$,
is equal to~$1$.  
(So we have now shown that the if branch in the for loop is correct.)
For components that contain more than one variable, it is
clear that we can choose one of those variables, $w_i$, and group the
set of $M$-partitions~$\sigma$ according to the value of~$\sigma(w_i)$.
(So we have now shown that the else branch is correct.)
Because there are no constraints between variables in different
components of the constraint graph, the number of satisfying
assignments factorises as $\prod_{i=1}^\kappa Z_i$.

For a binary relation $R$, we write
\begin{align*}
    \pi_1(R) &= \{a \mid (a,b)\in R \text{ for some }b\} \\
    \pi_2(R) &= \{b \mid (a,b)\in R \text{ for some }a\}\,.
\end{align*}

For 
the following
proof, we will also need the observation of Dyer and Richerby
\cite[Lemma~1]{DRfull} that any rectangular relation $R\subseteq
\pi_1(R)\times \pi_2(R)$ can be written as $(A_1\times B_1) \cup \dots
\cup (A_\lambda \times B_\lambda)$, where the $A_i$ and $B_i$
partition $\pi_1(R)$ and~$\pi_2(R)$, respectively.  The subrelations
$A_i\times B_i$ are referred to as \defn{blocks}.  A rectangular
relation $R\neq \pi_1(R)\times \pi_2(R)$ must have at least two
blocks.

\begin{lemma}\label{lem:quickarc}
Suppose that $\lists$ is subset-closed and $M$-purifying.
If  there is no \lists{}-$M$-derectangularising sequence, 
then Algorithm~\ref{alg:AC}
runs in polynomial time.
\end{lemma}
\begin{proof}
We will argue that the 
number of recursive calls made by
the function AC in Algorithm~\ref{alg:AC}
is bounded above by a polynomial in $|V|$.  This suffices, since every other step of the procedure is obviously polynomial.

Consider a run of the algorithm on instance $(V,C)$ which, by
Lemma~\ref{lemma:simple}, we may assume to be simple.  Suppose the
run makes a recursive call with input 
$(V_i,F'_{i,d})$.
For each $v\in V_i$, let $D'_{v}$ denote the arc-consistent domain for~$v$
that is computed during the recursive call.
We will show below that $D'_v\subset D_v$ for every 
variable~$v\in V_i$.  
This implies that the recursion depth is at most~$|D|$.
As a crude bound, it follows that the number of recursive calls is at most
${(|V|\cdot |D|)}^{|D|},$
since each recursive call 
that is made is nested below a sequence of at most~$|D|$ previous calls, each of
which chose a vertex $v\in V$ and ``pinned'' it to a domain element $d\in D$ (i.e., introduced the constraint $\const{v,\{d\}}$).

Towards showing that the domains of all variables decrease at each
recursive call, suppose that we are computing $\mathrm{AC}(V,C)$ and the
arc-consistent domains are $(D_v)_{v\in V}$.  As observed above, for
any component $H_i$ of the constraint graph on which a
recursive call is made, we must have $|D_v|>1$ for every $v\in V_i$.
Fix such a component and, for each $v\in V_i$, let $D'_v$ be the
arc-consistent domain calculated for~$v$ in the recursive call on
$H_i$.  It is clear that $D'_v\subseteq D_v$; we will show that $D'_v
\subset D_v$.

Consider a path $v_1\dots v_\ell$ in~$H_i$, where $v_1=w_i$
and $v_\ell=v$.  For each $j\in[\ell-1]$, there is exactly one binary
constraint in~$F_i$ involving $v_j$ and~$v_{j+1}$.  This is either
$\const{(v_j, v_{j+1}), R_j}$ or $\const{(v_{j+1}, v_j), R_j^{-1}}$
and, without loss of generality, we may assume that it is the former.
For $j\in[\ell-1]$, let $R'_j = R_j \cap (D_{v_j} \times D_{v_{j+1}}) =
H^M_{D_{v_j},D_{v_{j+1}}}$.
The relation~$R'_j$ is 
pure
because $D_{v_j}$
and $D_{v_{j+1}}$ are in the subset-closed set~$\lists$ 
and,
since $\lists$ is $M$-purifying, 
so is $\{D_{v_j},D_{v_{j+1}}\}$. 
These
two domains do not form a
derectangularising sequence by the hypothesis of the lemma, so 
$H^M_{D_{v_j},D_{v_{j+1}}}$ is rectangular.  
If some
$R_j=\emptyset$ then $D_{v_j} = D_{v_{j+1}} = \emptyset$ by
arc-consistency, contradicting the fact that $|D_v|>1$ for all $v\in
V_i$.  If some $R'_j$ has just one block, $R_j\cap (D_{v_j}\times
D_{v_{j+1}})$ is a Cartesian product, contradicting the fact that~$F$
is a factored set of constraints.  Thus, every $R'_j$ has at least two
blocks.

For $j\in[\ell-1]$, let $\Phi_j = R'_1 \circ \dots \circ R'_j$.
As above, note that $\{D_{v_1}, \ldots, D_{v_{j+1}}\}$ is $M$-purifying
and the sequence
$D_{v_1}, \dots, D_{v_{j+1}}$ is not derectangularising,
so $\Phi_j$ is rectangular.
We will show by
induction on~$j$ that $\pi_1(\Phi_j) = D_{v_1}$, $\pi_2(\Phi_j) = D_{v_{j+1}}$
and $\Phi_j$ has at least two blocks.  Therefore, since 
the recursive call constrains $\sigma(w_i)$ to be
$d$ and $d\in A$ for some block $A\times B\subset \Phi_\ell$, we have $D'_v\subseteq B\subset
D_v$, which is what we set out to prove.

For the base case of the induction, take $j=1$
so $\Phi_1=R'_1$. We showed above that $R'_1$ has at least two blocks 
and that $R'_1= H^M_{D_{v_1},D_{v_2}}$.
By arc-consistency, 
$\pi_1(R'_1) = D_{v_1}$ and
$\pi_2(R'_1) = D_{v_2}$.

For the inductive step, take $j\in [\ell-2]$.
Suppose that $\pi_1(\Phi_j)=D_{v_1}$,
$\pi_2(\Phi_j)=D_{v_{j+1}}$ and $\Phi_j = \bigcup_{s=1}^\lambda (A_s\times
A'_s)$ has at least two blocks.  We have $\Phi_{j+1} = \Phi_j\circ R'_{j+1}$
and $R'_{j+1} = \bigcup_{t=1}^\mu (B_t\times B'_t)$ for some $\mu \geq 2$.

For every $d\in D_{v_1}$, there is a $d'\in D_{v_{j+1}}$ such that
$(d,d')\in \Phi_j$ by the inductive hypothesis, and a $d''\in
D_{v_{j+1}}$ such that $(d'\!, d'')\in D_{v_{j+2}}$, by
arc-consistency.  Therefore, $\pi_1(\Phi_{j+1}) = D_{v_1}$; a similar
argument shows that $\pi_2(\Phi_{j+1}) = D_{v_{j+2}}$.

Suppose, towards a contradiction, that 
$\Phi_{j+1} = D_{v_1}\times
D_{v_{j+2}}$.  For this to be the case, we must have 
$A'_s\cap B_t\neq\emptyset$   for every $s\in\{1,2\}$ and $t\in[\mu]$.
Now, let $D^*_{v_{j+1}}=D_{v_{j+1}}\setminus
                            (A'_2\cap B_2)$ 
                                                     and
consider the relation
\begin{equation*}
    R = \{(d_1, d_3) \mid 
    \mbox{for some $d_2\in D^*_{v_{j+1}} $, 
                            $(d_1, d_2)\in \Phi_j$  and $(d_2, d_3)\in R'_{j+1}$ 
                            }\}.
          \end{equation*}
Since 
$A'_1 \subseteq D^*_{v_{j+1}}$
the 
non-empty
sets $A'_1 \cap B_1$ 
and $A'_1 \cap B_2$ are both subsets of
$D^*_{v_{j+1}}$  so
$A_1\times B'_1\subseteq R$ and
 $A_1\times B'_2\subseteq
R$.
Similarly, $B_1 \subseteq D^*_{v_{j+1}}$, so
$A'_2 \cap B_1 \subseteq D^*_{v_{j+1}}$ so
$A_2\times B'_1\subseteq R$. However, 
$(A_2\times B'_2)\cap R =
\emptyset$, so $R$ is not rectangular.
We will now derive a contradiction by showing that 
$R$ is rectangular.
Note that
$$R = H^M_{D_{v_1},D_{v_2}} \circ \cdots \circ H^M_{D_{v_{j-1}},D_{v_j}}
\circ
H^M_{D_{v_j},D^*_{v_{j+1}}}
\circ
H^M_{D^*_{v_{j+1}},D_{v_{j+2}}}$$
but this relation is rectangular because 
the hypothesis of the lemma guarantees that the sequence
$$D_{v_1},\ldots,D_{v_{j}},D^*_{v_{j+1}},D_{v_{j+2}}$$ 
is not an $\lists$-$M$-derectangularising sequence
and all of the elements of this sequence are in~$\lists$,
and  
$\{D_{v_1},\ldots,D_{v_{j}},D^*_{v_{j+1}},D_{v_{j+2}}\}$ is  
$M$-purifying.   \end{proof}

\section{Polynomial-time algorithms and the dichotomy theorem}
\label{sec:dichotomy}

Bulatov \cite{Bul08} showed that every problem of the form
$\nCSP(\Gamma)$ is either in $\FP$ or $\nP$-complete.   Together with 
Proposition~\ref{prop:purifiediscsp}, his result
immediately shows that a similar dichotomy exists for 
the special case of the
problem \LMPartitions{\mathcal L}{M} 
  in which
$\mathcal L$ is $M$-purifying and is closed under subsets.  
Our algorithmic work in Section~\ref{sec:arc}
can be combined with 
Dyer and Richerby's explicit dichotomy for $\nCSP$ to
obtain an explicit dichotomy for this special case of 
\LMPartitions{\mathcal L}{M}. 
In particular, Lemma~\ref{lem:quickarc}
gives a polynomial-time algorithm for the case in which there is no 
\lists{}-$M$-derectangularising sequence. When there is such a sequence,
$\GammaLM$ is not ``strongly rectangular'' in the sense
of~\cite{DRfull}.  It follows immediately that $\nCSP(\GammaLM)$ is
\numP{}-complete \cite[Lemma~24]{DRfull}  
so  \LMPartitions{\lists}{M} is  also \numP{}-complete by Proposition~\ref{prop:purifiediscsp}.
In fact, the dichotomy for this special case does not require the full generality
of Dyer and Richerby's dichotomy. If
there is an \lists{}-$M$-derectangularising sequence then it follows immediately from
work of Bulatov and Dalmau \cite[Theorem 2 and Corollary 3]{BD} that $\nCSP(\GammaLM)$ is
\numP{}-complete.

In this section we will
move beyond the case in which $\lists$ is $M$-purifying
to provide a full dichotomy for the problem
 \LMPartitions{\lists}{M}. We will use
two data structures:
 \emph{sparse-dense partitions}
and a representation of the set of \emph{splits} of a bipartite graph.
Similar data structures  were used by Hell et al.~\cite{HHN}
in their dichotomy for 
the \nPartitions{M} 
problem 
for matrices of size at most $3$-by-$3$.

\subsection{Data Structures}
\label{sec:DS}
 
We use two types of graph partition. 
The first is a special case
of a sparse-dense partition~\cite{FHKM}
which is also called an $(a,b)$-graph with $a=b=2$.
\begin{definition}
\label{def:bsd}
A bipartite--cobipartite partition 
of a graph~$G$ is
a partition~$(B,C)$ of~$V(G)$
such that $B$ induces a bipartite graph and $C$ induces the complement
of a bipartite graph.
\end{definition}
 
\begin{lemma}\label{lem:sparsedense}\cite[Theorem 3.1; see also the remarks on $(a,b)$-graphs.]{FHKM}
There is a polynomial-time algorithm for finding all 
bipartite--cobipartite partitions of a graph~$G$.
\end{lemma}

The second decomposition is based on 
certain sub-hypercubes called subcubes.
For any finite set $U\!$, a \defn{subcube} of $\{0,1\}^U$ is a 
subset of $\{0,1\}^U$  that is a
Cartesian
product of the form $\prod_{u\in U} S_u$ where
$S_u\in\{\{0\},\{1\},\{0,1\}\}$ for each $u\in U\!$.  We can also
associate a subcube $\prod_{u\in U} S_u$ with the set of assignments
$\sigma\colon U\to \{0,1\}$ such that $\sigma(u)\in S_u$ for all $u\in
U\!$.  Subcubes can be represented
efficiently by listing the projections~$S_u$.

\begin{definition}
    Let $G=(U,U'\!,E)$ be a bipartite graph, where $U$ and $U'$ are
    disjoint vertex sets, and $E\subseteq U\times U'\!$.  A
    \defn{subcube decomposition} of~$G$ is a list $U_1,\dots,U_k$ of
    subcubes of $\{0,1\}^U$ and a list $U'_1,\dots, U'_k$ of subcubes
    of $\{0,1\}^{U'}$ such that the following hold.
\begin{itemize}
\item The union $(U_1\times U'_1)\cup \dots \cup (U_k\times U'_k)$ is
  the set of assignments $\sigma\colon
  U\cup U'\to\{0,1\}$ such that:
  \begin{align}
   & \label{todayone} \mbox{no edge $(u,u')\in E$ has $\sigma(u)=\sigma(u')=0$ and}\\
   & \label{todaytwo} \mbox{no pair $(u,u')\in (U\times U')\setminus E$ has $\sigma(u)=
    \sigma(u')=1$.}
  \end{align}
\item For distinct $i,j\in[k]$, $U_i\times U'_i$ and $U_j\times U'_j$ 
are disjoint.
\item For each $i\in[k]$, either $|U_i|=1$ or $|U'_i|=1$ (or both).
\end{itemize}
\end{definition}

Note that, although we require $U_i\times U'_i$ and $U_j\times U'_j$
to be disjoint for distinct $i,j\in[k]$, we allow $U_i\cap
U_j\neq\emptyset$ as long as $U'_i$ and~$U'_j$ are disjoint, and
vice-versa.
It is even possible that $U_i=U_j$, and indeed this will happen in our constructions below.

\begin{lemma}\label{lem:splittocubes}
   A subcube decomposition of a bipartite graph $G=(U,U'\!,E)$ can be
   computed in polynomial time, with the subcubes represented by their
   projections.
\end{lemma}
\begin{proof}
For a vertex $x$ in a bipartite graph, let $\Gamma(x)$ be its set of
neighbours and let $\Gammabar(x)$ be its set of non-neighbours on the
other side of the graph.  Thus, for $x\in U\!$, $\Gammabar(x) =
U'\setminus \Gamma(x)$ and, for $x\in U'\!$, $\Gammabar(x) =
U\setminus \Gamma(x)$.

Observe that we can write $\{0,1\}^n\setminus
\{0\}^n$ as the disjoint union of $n$ subcubes $\{0\}^{k-1}\times
\{1\}^1\times \{0,1\}^{n-k}$ with $1\leq k\leq n$, and similarly for
any other cube minus a single point.

We first deal with two base cases.  If $G$ has no edges,
then the 
set of assignments $\sigma\colon  U\cup U'\to\{0,1\}$
satisfying (\ref{todayone}) and (\ref{todaytwo}) is the disjoint union of
\begin{equation*}
 \{0\}^U\times \{0\}^{U'}, \quad 
 (\{0,1\}^U\setminus\{0\}^U)\times \{0\}^{U'}, \quad \text{and} \quad
  \{0\}^U\times(\{0,1\}^{U'}\setminus\{0\}^{U'}).
  \end{equation*}
  The second and third terms can be decomposed into subcubes as described above to produce the output.
Similarly, if $G$ is is a complete bipartite graph, then the 
 set of assignments  
satisfying (\ref{todayone}) and (\ref{todaytwo}) 
is the disjoint union of
  \begin{equation*}
  \{1\}^U\times \{1\}^{U'}, \quad
  (\{0,1\}^U\setminus\{1\}^U)\times \{1\}^{U'}, \quad \text{and} \quad
  \{1\}^U\times(\{0,1\}^{U'}\setminus\{1\}^{U'}).
   \end{equation*}

If neither of these cases occurs then there is a vertex $x$ such that neither $\Gamma(x)$ nor
$\Gammabar(x)$ is empty.  If possible, choose $x\in U$; otherwise,
choose $x\in U'\!$.  To simplify the description of the algorithm, we
assume that $x\in U$; the other case is symmetric.  We consider
separately the assignments where $\sigma(x)=0$ and those where
$\sigma(x)=1$.  Note that, for any assignment, if $\sigma(y)=0$ for
some vertex~$y$, then $\sigma(z)=1$ for all $z\in\Gamma(y)$ and, if
$\sigma(y)=1$, then $\sigma(z)=0$ for all $z\in\Gammabar(y)$.
Applying this iteratively, setting $\sigma(x)=c$ for $c\in\{0,1\}$
also determines the value of $\sigma$ on some set $S_{x=c}\subseteq
U\cup U'$ of vertices.

Thus, we can compute a subcube decomposition for $G$ recursively.
First, compute $S_{x=0}$ and $S_{x=1}$.  Then, recursively compute
subcube decompositions of $G-S_{x=0}$ (the graph formed from~$G$ by
deleting the vertices in $S_{x=0}$) and $G-S_{x=1}$.  Translate these
subcube decompositions into a subcube decomposition of~$G$ by
extending each subcube $(U_i\times U'_i)$ of $G-S_{x=c}$ to a subcube
$(V_i\times V'_i)$ of $G$ whose restriction to $G-S_{x=c}$ is
$(U_i\times U'_i)$ and whose restriction to $S_{x=c}$ is an assignment
$\sigma$~with $\sigma(x)=c$ (in fact, all assignments that set $x$
to~$c$ agree on the set $S_{x=c}$, by construction).

It remains to show that the algorithm runs in polynomial time.  The
base cases are clearly computable in polynomial time, as are the
individual steps in the recursive cases, so we only need to show that the
number of recursive calls is polynomially bounded.  At the recursive
step, we only choose $x\in U'$ when $E(G) = U''\times U'$ for some
proper subset $\emptyset\subset U''\subset U$ and, in this case, the
two recursive calls are to base cases.  Since each recursive call when
$x\in U$ splits $U'$ into disjoint subsets, there can be at most
$|U'|-1$ such recursive calls, so the total number of recursive calls
is linear in $|V(G)|$.
\end{proof}

\subsection{Reduction to a problem with $M$-purifying lists}

Our algorithm for counting list $M$-partitions
uses the data structures from  Section~\ref{sec:DS}   to reduce problems where \lists{} is not $M$-purifying to
problems where it is (which we already know how to solve from Sections~\ref{sec:purifiedcsp} and~\ref{sec:arc}).
The algorithm is defined recursively
on the set $\lists$ of allowed
lists.  The algorithm for parameters~$\lists{}$ and $M$
calls the algorithm for~$\lists_i$ and $M$ 
where $\lists_i$ is a subset of~$\lists$.
The base case arises when $\lists_i$ is  
$M$-purifying.

We will use the following
computational problem to reduce \LMPartitions{\lists}{M} to a
collection of problems \LMPartitions{\lists'}{M} that are, in a sense,
disjoint.
\prob{ \LMPurify{\lists}{M}.}
{A graph $G$ and a
  function $L\colon V(G)\to\lists$.}
{Functions  $L_1,\dots,L_t\colon V(G)\to\lists$ such that
\begin{itemize}
\item for each 
$i\in[t]$, 
the set 
$\{L_i(v) \mid  v\in V(G)\}$ is $M$-purifying, 
\item for each $i\in [t]$ and $v \in V(G)$, $L_i(v) \subseteq L(v)$,
    and 
\item each
  $M$-partition of $G$ that respects $L$ respects exactly one of
  $L_1,\dots,L_t$.
  \end{itemize}
  }

We will give an algorithm for solving
the problem \LMPurify{\lists}{M} in polynomial time
when there is no \lists{}-$M$-derectangularising sequence of length
exactly~2. The following computational problem 
will be central to the inductive step. 

\prob{ \LMPurifyStep{\lists}{M}.}
{A graph $G$ and a function $L\colon V(G)\to \lists$.}
{Functions $L_1,\dots, L_k\colon
  V(G)\to \lists$ such that 
 \begin{itemize} 
\item for each $i\in [k]$ and $v \in V(G)$, $L_i(v) \subseteq L(v)$,
\item   every $M$-partition of~$G$ that respects $L$
  respects exactly one of $L_1,\dots,L_k$, and 
  \item for each $i\in[k]$,
there is a $W\in\lists{}$  
  which is inclusion-maximal in~\lists{}
but
does not occur in the image of $L_i$.
  \end{itemize}}

Note that we can trivially produce a solution to the problem
\LMPurifyStep{\lists}{M} by letting $L_1, \dots, L_k$ be an
enumeration of all possible functions 
such that all lists $L_i(v)$ have size~$1$ and satisfy  
$L_i(v) \subseteq L(v)$.
Such a function $L_i$ corresponds to an
assignment of vertices to parts so there is either exactly one
$L_i$-respecting $M$-partition or none, which means that every
$L$-respecting $M$-partition is $L_i$-respecting for exactly one~$i$.
However, this solution is exponentially large in $|V(G)|$ and we are
interested in solutions that can be produced in polynomial time.
Also, if $L(v)=\emptyset$ for some vertex~$v$, the algorithm is
entitled to output an empty list, since no $M$-partition
respects~$L$.

The following definition extends rectangularity to
$\{0,1,*\}$-matrices and is used in our proof.

\begin{definition}
A matrix $M\in\{0,1,*\}^{X\times Y}$ is
\defn{$*$-rectangular} if 
the relation $H^M_{X,Y}$ is rectangular.
\end{definition}
Thus, $M$ is $*$-rectangular if
and only if 
$M_{x,y}=M_{x'\!,y}=M_{x,y'}=*$ implies that $M_{x'\!,y'}=*$
for all 
$x,x'\in X'$ and all $y,y'\in Y''\!$.
 
We will show in Lemma~\ref{lem:claim}
that the function~\LMPurifyStep{\lists}{M} from
Algorithm~\ref{alg:purifystep} is a polynomial-time algorithm for 
the problem \LMPurifyStep{\lists}{M}
whenever \lists{} is not $M$-purifying and
 there is no length-2 \lists{}-$M$-derectangularising sequence. 
Note that a length-2
\lists{}-$M$-derectangularising sequence is a pair $X,Y\in\lists$ such
that $M|_{X\times Y}$, 
$M|_{X\times X}$ and 
$M|_{Y\times Y}$ 
are pure and $M|_{X\times Y}$
is 
not $*$-rectangular.   If $\lists\neq
\powerset{D}$, it is possible that a matrix that is not
$*$-rectangular has no length-2 \lists{}-$M$-derectangularising
sequence.  For example, let $D=\{1,2,3\}$ and $\lists =
\powerset{\{1,2\}}$ and let $M_{3,3}=0$ and $M_{i,j}=*$ for every other
pair $(i,j)\in D^2\!$.  $M$~is not $*$-rectangular but this fact is not
witnessed by any submatrix $M|_{X\times Y}$ for $X,Y\in\lists$.

\begin{algorithm}[t]
\caption{ 
A polynomial-time algorithm for the problem \LMPurifyStep{\lists}{M}
when $\lists{}\subseteq\powerset{D}$  is subset-closed,
   \lists{} is not $M$-purifying and
 there is no length-2 \lists{}-$M$-derectangularising sequence.
 The input is a pair $(G,L)$ with $V(G)=\{v_1,\ldots,v_n\}$. }
\label{alg:purifystep}
\begin{algorithmic}
\Function {\LMPurifyStep{\lists}{M}} {$G$,$L$}
\If{there is a $v_i\in V(G)$ with $L(v_i)=\emptyset$}
    \Return {the empty sequence} 
     \ElsIf{ there are $X,Y\in \lists$,  
   $a,b\in X$,  and $d\in Y$ 
 such that $M_{a,d}=0$ and $M_{b,d}=1$}
\State{Run Algorithm~\ref{alg:Case1} \quad /* Case 1 */}
\ElsIf{there is an $X\in \lists$ such that  $M|_{X\times X}$ is not pure}
\State{Run Algorithm~\ref{alg:Case2} \quad /* Case 2 */}
\Else 
\State{Run Algorithm~\ref{alg:Case3} \quad /* Case 3 */}
    \EndIf
    \EndFunction
\end{algorithmic}
\end{algorithm}

\begin{algorithm}[t]
\caption{Case 1 in Algorithm~\ref{alg:purifystep}.}
\label{alg:Case1}
\begin{algorithmic}
      \State{Choose 
      $X,Y\in \lists$,  
   $a,b\in X$,  and $d\in Y$ \\
 such that $M_{a,d}=0$, $M_{b,d}=1$ and $X$ and $Y$ are  inclusion-maximal in $\lists$}
      \For{$i\in[n]$}
        \State{$L_i(v_i)\gets
        L(v_i) \cap 
        \{d\}$}
        \For{$j<i$}
          \If{$(v_i,v_j)\in E(G)$}
             \State{$L_i(v_j)\gets\{d'\in L(v_j)\mid d'\neq d \text{ and } M_{d,d'}\neq 0\}$}
          \Else
             \State{$L_i(v_j)\gets\{d'\in L(v_j)\mid d'\neq d\text{ and } M_{d,d'}\neq 1\}$}             
                      \EndIf
        \EndFor
        \For{$j>i$}        
         \If{$(v_i,v_j)\in E(G)$}
              \State{$L_i(v_j)\gets\{ d'\in L(v_j)\mid M_{d,d'}\neq 0\}$}
          \Else
             \State{$L_i(v_j)\gets\{ d'\in L(v_j)\mid M_{d,d'}\neq 1\}$}
          \EndIf
        \EndFor
        \State{$L_{n+1}(v_i)\gets L(v_i)\setminus\{d\}$}
      \EndFor
      \Return {$L_1,\ldots,L_{n+1}$
(of course, if we have $L_i(v)=\emptyset$ for any $i$ and $v$ then $L_i$ can be omitted from the output)
      }
\end{algorithmic}
\end{algorithm}
 
\begin{algorithm}[t]
\caption{Case 2 in Algorithm~\ref{alg:purifystep}.}
\label{alg:Case2}
\begin{algorithmic}
\State{Choose $X\in \lists$ such that  $M|_{X\times X}$ is not pure and $X$ is inclusion-maximal in $\lists$}
      \State{Let $X_0\subseteq X$ be the set of rows of $M|_{X\times X}$ that contain a~$0$ }
      \State{$X_1 \gets X \setminus X_0$}
       \State{$V_X \gets \{v_j\in V(G)\mid L(v_j)=X\}$}
       \If{$V_X=\emptyset$}
           \Return {$L$}
        \Else
         \State{Use the algorithm promised in Lemma~\ref{lem:sparsedense} to compute the list $(B_1,C_1),\dots,(B_k,C_k)$ of all\\ bipartite--cobipartite partitions of $G[V_X]$}
         \For{$i\in[k],j\in [n]$} 
                \If{$v_j\notin V_X$}
                    \State{$L_i(v_j) \gets L(v_j)$}
                 \ElsIf{$v_j\in B_i$}
                        \State{$L_i(v_j)\gets X_0$}
                     \Else \quad /* $v_j\in C_i$*/
                       \State{$L_i(v_j)\gets X_1$}
                     \EndIf
            \EndFor 
  \Return {$L_1,\ldots,L_{k}$}
\EndIf
\end{algorithmic}
\end{algorithm}

\begin{algorithm}[t]
\caption{Case 3 in Algorithm~\ref{alg:purifystep}.}
\label{alg:Case3}
\begin{algorithmic}
     \State{Choose inclusion-maximal $X$ and $Y$ in \lists\ so that $M|_{X\times Y}$ is not pure}
     \State{Let $X_0\subseteq X$ be the set of rows of $M|_{X\times Y}$ that contain a~$0$ }
        \State{$X_1 \gets X \setminus X_0$}
             \State{Let $Y_0 \subseteq Y$ be the set of columns of $M|_{X\times Y}$ that contain a~$0$}
              \State{$Y_1 \gets Y \setminus Y_0$}  
     \State{$V_X \gets \{v_j\in V(G)\mid L(v_j)=X\}$} 
     \State{$V_Y \gets \{v_j\in V(G)\mid L(v_j)=Y\}$}  
      \If{$V_X=\emptyset$ or $V_Y=\emptyset$} \Return {$L$}     
      \Else
         \State{Let $E$ be the set of edges of $G$ between $V_X$ and $V_Y$}
         \State{Use the algorithm promised in Lemma~\ref{lem:splittocubes} to produce a subcube decomposition $(U_1,U'_1),\ldots,(U_k,U'_k)$ of $(V_X,V_Y,E)$}
           \For{$i\in[k],j\in [n]$} 
                 
                   \If{$v_j \in V_X$ and the projection of $U_i$ 
                   on
                   $v_j$ is $\{0\}$}
                       \State $L_i(v_j) \gets X_0$
                    \ElsIf{$v_j \in V_X$ and the projection of $U_i$ 
                    on $v_j$ is $\{1\}$}
                        \State $L_i(v_j) \gets X_1$ 
                                     
         \ElsIf{$v_j \in V_Y$ and the projection of $U'_i$ 
         on $v_j$ is $\{0\}$}
                       \State $L_i(v_j) \gets Y_0$
                    \ElsIf{$v_j \in V_Y$ and the projection of $U'_i$
                    on $v_j$ is $\{1\}$}
                        \State $L_i(v_j) \gets Y_1$

                   \Else
                         \State{$L_i(v_j) \gets L(v_j)$}
                   
                \EndIf      
           \EndFor
    \EndIf
          \Return{$L_1,\ldots,L_k$}
\end{algorithmic}
\end{algorithm}

 \begin{lemma}\label{lem:claim} 
 Let $M$ be a symmetric matrix in 
  $\{0,1,*\}^{D\times D}$ and
  let $\lists{}\subseteq\powerset{D}$ be subset-closed.
  If \lists{} is not $M$-purifying and
 there is no length-2 \lists{}-$M$-derectangularising sequence, then 
  Algorithm~\ref{alg:purifystep}  is a polynomial-time algorithm for
  the problem
\LMPurifyStep{\lists}{M}.
 \end{lemma}
 
\begin{proof}
We consider an
instance $(G,L)$ of the problem \LMPurifyStep{\lists}{M}
with $V(G)=\{v_1,\ldots,v_n\}$.
If there is a $v_i\in V(G)$ with $L(v_i)=\emptyset$ then  
no $M$-partition of~$G$ respects~$L$, so the output is correct.
Otherwise, we consider the three cases that can occur in the execution of the algorithm.

\paragraph{Case 1.} 

In this case column~$d$ of $M|_{X\times Y}$ contains both a zero and a one.
Equivalently, row~$d$ of $M|_{Y\times X}$ does.
Algorithm~\ref{alg:Case1} groups  the set of
$M$-partitions of~$G$ 
that respect $L$,
based on the first vertex that is placed in part~$d$.
For $i\in[n]$, $L_i$ requires that
$v_i$~is placed in part~$d$ and $v_1, \dots, v_{i-1}$ are not in
part~$d$; $L_{n+1}$ requires that part~$d$ is empty.  Thus,
no $M$-partition can  
respect
 more than one of
the~$L_i$.  Now consider an $L$-respecting $M$-partition $\sigma\colon
V(G)\to D$ and suppose that $i$~is minimal such that $\sigma(v_i)=d$.
We claim that $\sigma$~respects~$L_i$.  We have $\sigma(v_i)=d$, as
required.  For $j\neq i$, we must have $\sigma(v_j)\in L(v_j)$ since
$\sigma$~respects~$L$ and we must have $M_{d,\sigma(v_j)}\neq 1$ if
$(v_i, v_j)\notin E(G)$ and $M_{d,\sigma(v_j)}\neq 0$ if $(v_i,v_j)\in
E(G)$, since $\sigma$~is an $M$-partition.  In addition, by
construction, $\sigma(v_j)\neq d$ if $j<i$.  Therefore,
$\sigma$~respects~$L_i$.  A similar argument shows that $\sigma$
respects $L_{n+1}$ if $\sigma(v)\neq d$ for all $v\in V(G)$.  Hence,
any $M$-partition that respects~$L$ respects exactly one of the~$L_i$.

Finally, we show that, for each $i\in[n+1]$, there is a 
set~$W$  
which is
inclusion-maximal in~\lists{} and is not in the image of $L_i$.  For
$i\in [n]$, we cannot have both $a$ and~$b$ in $L_i(v_j)$ for any~$v_j$,
so~$X$ is not in the image of~$L_i$.  
$Y$~contains~$d$, so $Y$~is not
in the image of~$L_{n+1}$.

\paragraph{Case 2.} 
In this case, every row of 
   $M|_{X_0\times X}$ contains a~0,
while every 
row of $M|_{X_1\times X}$ 
fails to contain a zero.
Since $M|_{X\times X}$ is not
pure,  but no row of $M|_{X\times X}$ contains both 
a zero and a one
 (since we are not in Case~1),
$X_0$ and $X_1$ are non-empty.  Note that $M|_{X_0\times X_0}$ and
$M|_{X_1\times X_1}$ are both pure, while 
every entry of $M|_{X_0\times X_1}$  is a~$*$. 
 
If $V_X=\emptyset$ then $X$~is an inclusion-maximal
member of \lists{} that is not in the image of~$L$,
so the output of Algorithm~\ref{alg:Case2} is correct.
Otherwise, 
$(B_1,C_1),\dots,(B_k,C_k)$ is the list
containing all partitions $(B,C)$ of $V_X$
such that $B$ induces a bipartite graph in~$G$
and $C$ induces the complement of a bipartite graph.
The algorithm returns $L_1,\ldots, L_k$.
$X$ is not in the image of any $L_i$ so, 
to show that $\{L_1, \dots, L_k\}$ is a correct output for
the problem
\LMPurifyStep{\lists}{M}, we just need to show that every $M$-partition of $G$
that respects $L$ respects exactly one of $L_1,\dots,L_k$.
For $i\neq i'$, $(B_i,C_i)\neq (B_{i'},C_{i'})$ so there is at least one
vertex~$v_j$ such that $L_i(v_j)=X_0$ and $L_{i'}(v_j)=X_1$ or vice-versa.
Since $X_0$ and~$X_1$ are disjoint, no $M$-partition can
simultaneously respect $L_i$ and $L_{i'}$.  It remains to show that every
$M$-partition respects at least one of $L_1, \dots, L_k$.
To do this, we deduce two structural properties 
of $M|_{X\times X}$.

First, we show that $M|_{X\times X}$ has no $*$ on its diagonal.
Suppose towards a contradiction that $M_{d,d}=*$ for some $d\in X$.
If $d\in X_0$, then, for each $d'\in X_1$, $M_{d,d'}=M_{d'\!,d}=*$
because, as noted above, every entry of $M|_{X_0\times X_1}$ is a~$*$.
Therefore, the $2\times 2$ matrix $M'=M|_{\{d,d'\}\times \{d,d'\}}$
contains at least three~$*$s so it is pure.  $\{d,d'\}
\subseteq
X\in
\lists$ so, by the hypothesis of the lemma, the length-2 sequence
$\{d,d'\},\{d,d'\}$ is not \lists{}-$M$-derectangularising, so $M'$
must be $*$-rectangular, so $M_{d'\!,d'}=*$ for all $d'\in X_1$.
Similarly, if $M_{d'\!,d'}=*$ for some $d'\in X_1$, then $M_{d,d}=*$
for all $d\in X_0$.  Therefore, if $M|_{X\times X}$ has a~$*$ on its
diagonal, every entry on the diagonal is~$*$.  But $M$ contains a~0,
say $M_{i,j}=0$ with $i,j\in X_0$.  For any $k\in X_1$,
\begin{equation*}
    M|_{\{i,j\}\times \{j,k\}} = \begin{pmatrix} 0 & * \\ * & * \end{pmatrix},
\end{equation*}
so the length-2 sequence $\{i,j\}, \{j,k\}$ is
\lists{}-$M$-derectangularising, contradicting the hypothesis of the
lemma (note that $\{i,j\},\{j,k\}\subseteq X\in\lists$).

Second, we show that  there is no sequence
$d_1,\dots,d_\ell\in X_0$ of odd length such that
\begin{equation*}
    M_{d_1,d_2}=M_{d_2,d_3}=\dots=M_{d_{\ell-1},d_\ell}=M_{d_\ell,d_1}=*\,.
\end{equation*}
Suppose for a contradiction that such a sequence exists. 
Note that $M|_{X_0\times X_0}$ is
$*$-rectangular 
since $X_0,X_0$ is not an \lists{}-$M$-derectangularising sequence and 
$M|_{X_0\times X_0}$ is pure since
Case~1 does not apply.
We will show by induction that for every  non-negative integer $\kappa \leq (\ell-3)/2$,
$M_{d_1,d_{\ell-2\kappa-2}}=*$.
This gives a contradiction by taking $\kappa=(\ell-3)/2$ since 
$M_{d_1,d_1}=*$ and we have already shown that $M|_{X_0\times X_0}$
has no $*$ on its diagonal. 
For every~$\kappa$, the argument follows by considering the
matrix $M_\kappa = M|_{\{d_1,d_{\ell-2\kappa-1}\} \times \{d_{\ell-2\kappa-2},d_{\ell-2\kappa}\}}$.
The definition of the sequence $d_1,\ldots,d_\ell$ together with the symmetry of~$M$
guarantees that both entries in row $d_{\ell-2\kappa-1}$ of~$M_\kappa$ are equal to $*$.
It is also true that $M_{d_1,d_{\ell-2\kappa}}=*$: If $\kappa=0$ then this follows
from the definition of the sequence; otherwise it follows by induction.
The fact that $M_{d_1,d_{\ell-2\kappa-2}}=*$ then follows by $*$-rectangularity.
 
This second structural property implies that, for any 
$M|_{X\times X}$-partition of $G[V_X]$, the
graph induced by vertices assigned to $X_0$ has no odd cycles,
and is therefore bipartite. Similarly, the
vertices assigned to $X_1$ induce the complement of a bipartite graph.
Therefore, any $M$-partition of~$G$ that respects~$L$ must respect  at least
one of the $L_1, \dots, L_k$,
so it respects exactly one of them, as required.

\paragraph{Case 3.} 
Since Cases 1 and~2 do not apply and \lists{} is not $M$-purifying, 
there are  
distinct
$X,Y\in
\lists$ such that 
$X$ and $Y$ are inclusion-maximal in \lists\ and
 $M|_{X\times Y}$ is not pure. As in the
previous case, the sets $X_0$, $X_1$, $Y_0$ and~$Y_1$ are all non-empty.
 
If either $V_X$ or $V_Y$ is empty then 
either $X$ or~$Y$ is
an inclusion-maximal set in~\lists{} that is not in the image of~$L$
so the output of Algorithm~\ref{alg:Case3} is correct. Otherwise,
 $(U_1,U'_1),\dots,(U_k,U'_k)$
 is a subcube decomposition of  the bipartite subgraph $(V_X,V_Y,E)$. 
  The $U_i$s are subcubes of~$\{0,1\}^{V_X}$ and the 
$U'_i$s are subcubes of~$\{0,1\}^{V_Y}$.
The algorithm returns $L_1,\ldots,L_k$.

Note that if 
$|U'_i|=1$ then 
$Y$ is not in the image of~$L_i$.
Similarly, if $|U'_i|>1$ but
$|U_i|=1$ then
$X$ is not in the image of~$L_i$.
The definition of  subcube decompositions guarantees that, for every~$i$, at least one of these is the case.
To show this definition of $L_1,\ldots,L_k$ is a correct
 output for the problem \LMPurifyStep{\lists}{M}, we must show that any
$M$-partition of $G$ that respects $L$ also respects exactly one $L_i$.
Since   the sets in 
$\{U_i \times U'_i \mid  i\in[k]\}$ are disjoint subsets of $\{0,1\}^{V_X\cup V_Y}$, 
any $M$-partition of $G$ that respects $L$ respects at most one~$L_i$
so it remains to show that every $M$-partition of~$G$ respects at
least one $L_i$.
To do this, we deduce two structural properties of $M|_{X\times Y}$.

First, we show that every entry of $M|_{X_0\times Y_0}$ is $0$.
The definition of~$X_0$ guarantees that every row of
$M|_{X_0\times Y_0}$ contains a~$0$.
Since Case~1 does not apply, and $M$ is symmetric, 
every entry of $M|_{X_0\times Y_0}$ is either $0$ or $*$.
Suppose for a contradiction that $M_{i,j}=*$ for some $(i,j)\in
X_0\times Y_0$.  Pick $i'\in X_1$. For any $j'\in Y_0\setminus\{j\}$
we have $M_{i,j}=M_{i'\!,j}=M_{i'\!,j'}=*$, so by $*$-rectangularity of
$M|_{X\times Y_0}$ we have $M_{i,j'}=*$.  
Thus,  every entry of  $M|_{\{i\}\times Y_0}$  is $*$,
so there is a $*$ in every $Y_0$-indexed column of~$M$. By the same
argument, swapping the roles of $X$ and $Y$, every entry in 
$M|_{X_0\times Y_0}$ is $*$,  contradicting the fact that $M|_{X\times Y}$ contains a
$0$ since $M|_{X\times Y}$ is not pure.

Second, a similar argument shows 
that every entry of $M|_{X_1\times Y_1}$ is $1$.

Thus for all $M$-partitions $\sigma$ of $G$ respecting $L$, for all
$x\in V_X$ and $y\in V_Y$, if $(x,y)\in E$ then
$(\sigma(x),\sigma(y))\notin X_0\times Y_0$ while if $(x,y)\notin E$ then
$(\sigma(x),\sigma(y))\notin X_1\times Y_1$.  
Using the definition of subcube decompositions,
this shows that any
$M$-partition of $G$ respecting $L$ respects some $L_i$.
\end{proof}

We can now give an algorithm for  the problem \LMPurify{\lists}{M}.
The algorithm consists of the function \LMPurify{\lists}{M}, which
is defined in Algorithm~\ref{alg:purifytriv} 
for the trivial case in which $\lists$ is $M$-purifying
and in Algorithm~\ref{alg:purify} for the  
case in which it is not.
Note that for any fixed~$\lists$ and~$M$ the algorithm is
defined either in Algorithm~\ref{alg:purifytriv} or in Algorithm~\ref{alg:purify}
and the function \LMPurify{\lists}{M}
 is not recursive. However, the \emph{definition}
is recursive, so the function \LMPurify{\lists}{M} 
defined in Algorithm~\ref{alg:purify} does
make a call to a function \LMPurify{\lists_i}{M}
for some $\lists_i$ which is smaller than $\lists$.
The function \LMPurify{\lists_i}{M}
is in turn defined in Algorithm~\ref{alg:purifytriv} or~Algorithm~\ref{alg:purify}.
\begin{algorithm}[t]
\caption{A trivial algorithm for the problem \LMPurify{\lists}{M} for
the case in which $\lists$ is $M$-purifying.}
\label{alg:purifytriv}
\begin{algorithmic}
\Function{\LMPurify{\lists}{M}}{$G$,$L$}
\Return{$L$}
\EndFunction
\end{algorithmic}
\end{algorithm}
\begin{algorithm}[t]
\caption{A polynomial-time algorithm for  the problem \LMPurify{\lists}{M}
when $\lists\subseteq\powerset{D}$ is subset-closed and 
is not $M$-purifying and
there is no length-$2$ \lists{}-$M$-derectangularising sequence.
This algorithm calls the function \LMPurifyStep{\lists}{M} from Algorithm~\ref{alg:purifystep}.
It also calls the function \LMPurify{\lists_i}{M}
for various lists~$\lists_i$ which are shorter than~$\lists$.
These functions are defined inductively in Algorithm~\ref{alg:purifytriv} and here.}
\label{alg:purify}
\begin{algorithmic}
\Function{\LMPurify{\lists}{M}}{$G$,$L$}
  \State{ /* $\emptyset \in \lists$ since $\lists$ is subset-closed.
   Since $\lists$ is not $M$-purifying, $\lists \neq \{\emptyset\}$, hence
      $|\lists|>1$ */}
   \State{Let $B$ be the empty sequence of list functions}
   \State{$L_1,\ldots,L_k \gets$ \LMPurifyStep{\lists}{M}$(G,L)$}
       \For{$i\in[k]$}  
            \State{$\lists_i \gets {\bigcup_{v\in V(G)} \powerset{L_i(v)}}$}   
            \State{ $L'_1,\ldots,L'_j \gets$ \LMPurify{\lists_i}{M}$(G,L_i)$}
             \State{Add $L'_1,\ldots,L'_j$ to $B$}
     \EndFor
    \Return{$B$}
\EndFunction
\end{algorithmic}
\end{algorithm}
The correctness of  the algorithm follows from the definition of
the problem.
The following lemma 
bounds the running time. 

\begin{lemma}\label{lem:reducetopure}
Let  $M\in\{0,1,*\}^{D\times D}$ 
be a symmetric matrix and let 
 $\lists\subseteq\powerset{D}$ be subset-closed.   
If there is no length-$2$ \lists{}-$M$-derectangularising sequence,
then the function~\LMPurify{\lists}{M} 
as defined in Algorithms~\ref{alg:purifytriv} and~\ref{alg:purify}
is a polynomial-time algorithm for  the problem \LMPurify{\lists}{M}.
\end{lemma}
\begin{proof}

Note that $\lists$ is a fixed parameter of the problem \LMPurify{\lists}{M} ---
it is not part of the input.
The proof is by induction on $|\lists|$.
If $|\lists|=1$ then 
$\lists=\{\emptyset\}$ so it is $M$-purifying.
In this case,
function~\LMPurify{\lists}{M} 
is defined in Algorithm~\ref{alg:purifytriv}. It is clear that it is a polynomial-time
algorithm for the problem \LMPurify{\lists}{M}.

For the inductive step suppose that $|\lists|>1$.
If $\lists$ is $M$-purifying then 
function~\LMPurify{\lists}{M} 
is defined in Algorithm~\ref{alg:purifytriv}
and again the result is trivial. Otherwise,
function~\LMPurify{\lists}{M} 
is defined in Algorithm~\ref{alg:purify}.
Note that 
$\lists\subseteq\powerset{D}$ is subset-closed and
there is no length-$2$ \lists{}-$M$-derectangularising sequence. 
From this, we can conclude that,
for any subset-closed subset $\lists'$ of~$\lists$,
there is no length-$2$ $\lists'$-$M$-derectangularising sequence.
So we can assume by the inductive hypothesis that for all 
subset-closed $\lists'\subset \lists{}$, 
the function 
\LMPurify{\lists'}{M} runs in polynomial time.

The result now follows from the fact that 
the function
\LMPurifyStep{\lists}{M} runs in polynomial time (as guaranteed by Lemma~\ref{lem:claim})
and from the fact 
that each $\lists_i$ is a strict subset of~$\lists$,
which follows from the definition of 
problem \LMPurifyStep{\lists}{M}.
Each $M$-partition that respects $L$ respects exactly one of
$L_1,\dots,L_k$ and, hence, 
it respects exactly one of the list functions that is returned.
\end{proof}

\subsection{Algorithm for \LMPartitions{\lists}{M} and proof of the dichotomy}

We can now present our algorithm for 
the problem \LMPartitions{\lists}{M}.
The algorithm consists of the
function \LMPartitions{\lists}{M} 
which is defined in Algorithm~\ref{alg:mainpurifying}
for the   case in which $\lists$ is $M$-purifying
and in Algorithm~\ref{alg:main} when it is not.
\begin{algorithm}[t]
\caption{A polynomial-time algorithm for the problem \LMPartitions{\lists}{M}
when $\lists$ is subset-closed and
$M$-purifying and there is no $\lists$-$M$-derectangularising sequence.}
\label{alg:mainpurifying}
\begin{algorithmic}
\Function{\LMPartitions{\lists}{M}}{$G$,$L$}
\State{$(V,C) \gets$ the instance of $\nCSP(\GammaLM)$
obtained by applying the polynomial-time Turing reduction from Proposition~\ref{prop:purifiediscsp} to the input $(G,L)$}

\Return AC$(V,C)$ where AC is the function from Algorithm~\ref{alg:AC}
\EndFunction
\end{algorithmic}
\end{algorithm}
\begin{algorithm}[t]
\caption{A polynomial-time algorithm for the problem
\LMPartitions{\lists}{M}
when $\lists$ is subset-closed and
not $M$-purifying and
there is no $\lists$-$M$-derectangularising sequence.
The algorithm calls the function {\LMPurify{\lists}{M}}{$(G,L)$}
from 
Algorithm~\ref{alg:purify}.
}
\label{alg:main}
\begin{algorithmic}
\Function{\LMPartitions{\lists}{M}}{$G$,$L$}
\State{$L_1,\ldots,L_t \gets$ {\LMPurify{\lists}{M}}{$(G,L)$}}
\State{$Z\gets 0$}
\For{$i\in[t]$}
       \State{$\lists_i \gets \bigcup_{v \in V(G)} \powerset{L_i(v)}$}   
\State{$(V,C_i) \gets$ the instance of 
$\nCSP(\Gamma_{\lists_i,M})$
obtained by applying the polynomial-time Turing reduction from Proposition~\ref{prop:purifiediscsp} to the input $(G,L_i)$}

\State{ $Z_i \gets \text{AC$(V,C_i)$ where AC is the function from Algorithm~\ref{alg:AC}}$}
\State{$Z \gets Z + Z_i$}    
\EndFor
\Return{$Z$} 
\EndFunction
\end{algorithmic}
\end{algorithm}

\begin{lemma}
\label{lem:positive} 
Let  $M\in\{0,1,*\}^{D\times D}$ 
be a symmetric matrix and let 
 $\lists\subseteq\powerset{D}$ be subset-closed.   
If there is no   \lists{}-$M$-derectangularising sequence,
then 
the function~\LMPartitions{\lists}{M} 
as defined in Algorithms~\ref{alg:mainpurifying} and~\ref{alg:main}
is a polynomial-time algorithm for  the problem \LMPartitions{\lists}{M}.\end{lemma}

\begin{proof}
If $\lists$ is $M$-purifying then 
the function~\LMPartitions{\lists}{M} 
is defined in Algorithm~\ref{alg:mainpurifying}.
Proposition~\ref{prop:purifiediscsp} shows that the   reduction 
in Algorithm~\ref{alg:mainpurifying} to a CSP instance is correct and takes polynomial time.
The CSP instance can be solved by the function AC in Algorithm~\ref{alg:AC}, whose running time 
is shown to be polynomial in  Lemma~\ref{lem:quickarc}. 

If $\lists$ is not $M$-purifying then 
the function \LMPartitions{\lists}{M} 
is defined in Algorithm~\ref{alg:main}.
Lemma~\ref{lem:reducetopure} guarantees
that the function \LMPurify{\lists}{M}
is a polynomial-time algorithm for the problem~\LMPurify{\lists}{M}. 
If the list $L_1,\ldots,L_t$ is empty
then there is no $M$-partition of~$G$ that respects~$L$
so it is correct that the function \LMPartitions{\lists}{M}
returns~$0$.
Otherwise, we know from the definition of the problem \LMPurify{\lists}{M} that 
\begin{itemize}
\item functions $L_1,\ldots,L_t$ are from~$V(G)$ to~$\lists$,
\item for each 
$i\in [t]$, 
the set 
$\{L_i(v) \mid  v\in V(G)\}$ is $M$-purifying,
\item for each $i\in [t]$ and $v \in V(G)$, $L_i(v) \subseteq L(v)$,
    and 
\item each
  $M$-partition of $G$ that respects $L$ respects exactly one of
  $L_1,\dots,L_t$.
\end{itemize}
The desired result is now
the sum, over all~$i\in[t]$,
of the number of $M$-partitions of~$G$ that respect~$L_i$.
Since the list $L_1, \dots, L_t$ is generated in polynomial time,
$t$~is bounded by some polynomial in $|V(G)|$.  

Now, for each $i\in[t]$, $\lists_i$ is a subset-closed
subset of $\lists$. Since there is no $\lists$-$M$-derectangularising sequence,
there is also no $\lists_i$-$M$-derectangularising sequence.
Also,
$\lists_i$ is $M$-purifying.
Thus, the argument that we gave for the purifying case shows that $Z_i$ is the desired quantity.
\end{proof}

We can now combine our results to 
establish our dichotomy for the problem \LMPartitions{\lists}{M}.

\begin{reptheorem}{thm:explicitdichotomy}\explicitdichotomy\end{reptheorem}
\begin{proof}

Suppose that there is an \lists{}-$M$-derectangularising sequence
$D_1,\dots,D_k$.
Recall (from Definition~\ref{def:closure}) the definition of
the subset-closure $\subclo{\lists''}$ of a set $\lists'' \subseteq \powerset{D}$.
Let 
$$\lists'= \subclo {\{D_1,\ldots,D_k\} }.$$
Since $\{D_1,\ldots,D_k\}$ is $M$-purifying, so is $\lists'\!$, which is
also subset-closed.
It follows that $\Gamma_{\!\lists'\!,M}$ is well defined 
(see Definition~\ref{defgammaprime})
and contains
the relations $H_{D_1,D_2}^M, \ldots,H_{D_{k-1},D_k}^M$ (and possibly others).
Since $H_{D_1,D_2}^M \circ H_{D_2,D_3}^M \circ \cdots \circ H_{D_{k-1},D_k}^M$ is
not rectangular,  $\nCSP(\Gamma_{\!\lists'\!,M})$ is $\numP$-complete 
\cite[Theorem~2 and Corollary~3]{BD} 
(see also \cite[Lemma~24]{DRfull}).
By Proposition~\ref{prop:purifiediscsp},  the problem
\LMPartitions{\lists'}{M} is $\numP$-complete
so the more general problem \LMPartitions{\lists}{M} is also $\numP$-complete.
On the other hand, if there is no $\lists$-$M$-derectangularising sequence,
then the result follows from Lemma~\ref{lem:positive}.\end{proof}

\section{Complexity of the dichotomy criterion}
\label{sec:meta}

The dichotomy established in Theorem~\ref{thm:explicitdichotomy} is that,
if there is an
\lists{}-$M$-derectangularising sequence, then the problem
\LMPartitions{\lists}{M} is $\numP$-complete; otherwise,
 it is in $\FP$.  
 This section addresses the 
computational problem of determining which is the case, given \lists{} and~$M$.

The following lemma 
will allow us to show
that the problem \ExistsDerectSeq\ 
(the problem of determining whether there is an $\subclo{\lists}$-$M$-derectangularising sequence,
given \lists{} and $M$)
and the related problem \MatrixHasDerectSeq\ 
(the problem of determining whether there is a $\powerset{D}$-$M$-derectangularising sequence,
given $M$)
are both in \NP{}.  Note that, for this ``meta-problem'', \lists{}
and~$M$ are the inputs whereas, previously, we have regarded them as
fixed parameters.

\begin{lemma}\label{lem:small_derect}
Let $M\in\{0,1,*\}^{D\times D}$ be symmetric, and let
$\lists\subseteq  \powerset{D}$ be subset-closed.
If there is an \lists{}-$M$-derectangularising sequence, then there
is one of length at most
$512(|D|^3+1)$.
\end{lemma}
\begin{proof}
Pick an \lists{}-$M$-derectangularising sequence $D_1,\dots,D_k$
with $k$ minimal; we will show that $k\leq 512(|D|^3+1)$.
Define $$R=H^M_{D_1, D_2} \circ H^M_{D_2, D_3} \circ \dots \circ
H^M_{D_{k-1}, D_k}.$$ Note that $R\subseteq D_1\times D_k$.  By the definition
of derectangularising sequence, there
are 
$a,a'\in D_1$ and $b,b'\in D_k$  such that $(a,b)$, $(a'\!,b)$ and $(a,b')$
are all in~$R$ but
$(a'\!,b')\not\in R$.  So there exist
\begin{equation*}
    (x_1,\dots,x_k),(y_1,\dots,y_k),(z_1,\dots,z_k)\in D_1\times \dots
    \times D_k
\end{equation*}
with $(x_1,x_k)=(a,b)$, $(y_1,y_k)=(a'\!,b)$ and
$(z_1,z_k)=(a,b')$ such that $M_{x_i,x_{i+1}} = M_{y_i,y_{i+1}} =
M_{z_i,z_{i+1}}=*$ for every $i\in[k-1]$ but, for any $(w_1,\ldots,w_k)\in D_1 \times \dots \times D_k$
with $(w_1,w_k)=(a'\!,b')$, there is an~$i\in[k-1]$ such that
$M_{w_i,w_{i+1}}\neq *$.  

Setting $D'_i=\{x_i,y_i,z_i\}$
for each $i$ gives an \lists{}-$M$-derectangularising sequence
$D'_1,\dots,D'_k$ with $|D'_i|\leq 3$ for each $1\leq i\leq k$.  (Note
that any submatrix of a pure matrix is pure.)  For all $1\leq s <
t\leq k$ define
\[ R_{s,t}=H^M_{D'_s, D'_{s+1}} \circ H^M_{D'_{s+1}, D'_{s+2}} \circ \dots \circ
  H^M_{D'_{t-1}, D'_t}. \]
Since $D'_1,\ldots,D'_k$ is \lists{}-$M$-derectangularising,
$R_{1,k}$ is not rectangular but, by the minimality of $k$, every other
$R_{s,t}$ is rectangular.  Note also that no $R_{s,t}=\emptyset$
since, if that were the case, we would have $R_{1,k}=\emptyset$, which
is rectangular.

Suppose for a contradiction that $k> 512(|D|^3+1)$.  There are at most
$|D|^3+1$ subsets of $D$ with size at most three, so there are indices
$1\leq i_0<i_1<i_2<\dots<i_{512}\leq k$ such that
$D'_{i_0}=\dots=D'_{i_{512}}$.  There are at most 
$2^{|D'_{i_0}|^2}-1 \leq 2^9-1=511$ non-empty binary
relations on $D'_{i_0}$, so $R_{i_0,i_m}=R_{i_0,i_n}$ for some $1\leq 
m<n\leq 512$.  Since $R_{1,k}$ is not rectangular,
\begin{equation*}
     R_{1,k}= R_{1,i_0} \circ R_{i_0,i_n} \circ R_{i_n,k}= R_{1,i_0} \circ R_{i_0,i_m} \circ R_{i_n,k}= R_{1,i_m} \circ R_{i_n,k}
\end{equation*}
is not rectangular. Therefore,
$D'_1,D'_2,\dots,D'_{i_m},D'_{1+i_n},D'_{2+i_n},\dots,D'_k$ is an $\mathcal
L$-$M$-derectangularising sequence of length less than $k$,
contradicting the minimality of $k$.
\end{proof}

Now that we have membership in \NP{}, we can prove completeness.

 \begin{reptheorem}{thm:meta}\metathm\end{reptheorem}
\begin{proof}

We first show that \ExistsDerectSeq\ is in \NP. 
Given $D$, $M\in \{0,1,*\}^{D\times D}$ and
$\lists \subseteq \powerset{D}$, 
a non-deterministic polynomial time algorithm for \ExistsDerectSeq\
first ``guesses'' an $\subclo{\lists}$-$M$-derectangularising sequence $D_1,\ldots,D_k$
with $k\leq 512{(|D|^3+1)}$.
Lemma~\ref{lem:small_derect} 
guarantees that such a sequence exists 
if the output should be ``yes''.
The algorithm then verifies that 
each $D_i$ is a subset of a set in $\lists$,
that $\{D_1,\ldots,D_k\}$ is $M$-purifying,
and that the
relation $H^M_{D_1,D_2} \circ H^M_{D_2, D_3} \circ \dots \circ
H^M_{D_{k-1}, D_k}$ is not rectangular.
All of these can be checked in polynomial time without
explicitly constructing $\subclo{\lists}$.

To show that \ExistsDerectSeq\ is \NP{}-hard, we give
a polynomial-time reduction from the well-known \NP{}-hard problem of determining whether
a graph~$G$ has an independent set of size~$k$.
 
Let $G$ and $k$ be an input to the independent set problem.  Let
$V(G)= [n]$ and assume without loss of generality that
$k\in[n]$.  Setting $D=[n]\times[k]\times[3]$, we construct a $D\times
D$ matrix~$M$ and a set \lists{} of lists such that there is an
$\subclo{\lists}$-$M$-derectangularising sequence if and only if $G$~has an
independent set of size~$k$.

$M$ will be a block matrix, constructed using the following $3\times
3$ symmetric matrices.  Note that each is pure, apart from~$\Mid$.
\begin{gather*}
    \Mstart = \begin{pmatrix}*&*&0 \\ *&*&0 \\ 0&0&*\end{pmatrix} \qquad
    \Mend   = \begin{pmatrix}*&0&0 \\ 0&*&* \\ 0&*&*\end{pmatrix} \qquad
    \Mbij   = \begin{pmatrix}*&0&0 \\ 0&*&0 \\ 0&0&*\end{pmatrix} \\
    \Mzero  = \begin{pmatrix}0&0&0 \\ 0&0&0 \\ 0&0&0\end{pmatrix} \qquad
    \Mid    = \begin{pmatrix}1&0&0 \\ 0&1&0 \\ 0&0&1\end{pmatrix}\,.
\end{gather*}

For $v\in [n]$ and $j\in[k]$, let $D[v,j] = \{(v,j,c) \mid c\in[3]\}$.
Below, when we say that $M|_{D[v,j]\times D[v',j']}= N$ for some
$3\times 3$ matrix $N$, we mean more specifically that
$M_{(v,j,c),(v'\!,j'\!,c')} = N_{c,c'}$ for all $c,c'\in[3]$.  $M$~is
constructed as follows.
\begin{itemize}
\item For all $v\in[n]$, $M|_{D[v,1] \times D[v,1]}= \Mstart$ and 
$M|_{D[v,k]\times D[v,k]} =
    \Mend$.
\item For all $v\in[n]$ and all $j\in\{2,\dots,k-1\}$, $M|_{D[{v,j}]\times D[{v,j}]} =
    \Mbij$.
\item If $v\neq v'\!$, $(v,v')\notin E(G)$ and $j<k$, then
    \begin{itemize}
    \item $M|_{D[{v,j}]\times D[{v',j+1}]}   = M|_{D[v',j+1]\times D[{v,j}]}  = \Mbij$ and
    \item $M|_{D[{v,j}] \times D[{v',j'}]}  = M|_{D[{v',j'}]\times D[{v,j}]}   = \Mzero$ for all
$j'>j+1$.
     \end{itemize}
\item For all $v,v'\in[n]$ and $j,j'\in[k]$ not covered above,
  $M|_{D[v,j]\times D[v',j']}  = \Mid$.
\end{itemize}
To complete the construction, let 
$\lists=\{D[v,j]\mid v\in[n], j\in[k]\}$.  
We will show that 
$G$ has an independent set of size~$k$ if and only if
there is an $\subclo{\lists}$-$M$-derectangularising sequence.
 
For the forward direction of the proof, 
suppose that $G$ has an independent set $I = \{v_1, \dots, v_k\}$ of
size~$k$.  We will show that
\begin{equation*}
    D[v_1,1], D[v_1,1], D[v_2,2], D[v_3,3], \dots, D[v_{k-1}, k-1],
    D[v_k,k], D[v_k,k]
\end{equation*}
(where the first and last elements are repeated and the others are
not) is $\subclo{\lists}$-$M$-derectangularising.  Since there is no edge
$(v_i,v_{i'})\in E(G)$ for $i,i'\in[k]$, the matrix $M|_{D[v_i,i]\times D[v_{i'},i']}  $
is always one of $\Mstart$, $\Mend$, $\Mbij$ and $\Mzero$, so it is
always pure.  Therefore, $\{D[v_1,1], \dots, D[v_k,k]\}$ is
$M$-purifying.  It remains to show that the relation
\begin{equation*}
    R = H^M_{D[v_1,1],D[v_1,1]} \circ H^M_{D[v_1,1],D[v_2,2]} \circ \dots
            \circ H^M_{D[v_{k-1,k-1}],D[v_{k},k]} \circ H^M_{D[v_k,k],D[v_k,k]}
\end{equation*}
is not rectangular.

Consider $i\in[k-1]$. Since $(v_i,v_{i+1})\notin E(G)$, $M|_{ D[{v_i,i}] \times D[{ v_{i+1}, i+1 }]}  = \Mbij$ so
$H^M_{D[v_i,i], D[v_{i+1},i+1]}$ is the bijection that associates
$(v_i,i,c)$ with $(v_{i+1},i+1,c)$ for each $c\in[3]$.  Therefore,
$$H^M_{D[v_1,1],D[v_1,2]} \circ \dots \circ
H^M_{D[v_{k-1},k-1],D[v_k,k]}$$ is the bijection that associates
$(v_1,1,c)$ with $(v_k,k,c)$ for each $c\in[3]$.  We have $M|_{D[v_1,1]\times D[v_1,1]} 
= \Mstart$ and $M|_{D[v_k,k] \times D[v_k,k]}  = \Mend$ so
\begin{align*}
    H^M_{D[v_1,1], D[v_1,1]} &= \{((v_1,1,c),(v_1,1,c')) \mid c,c'\in[2]\}
                        \cup \{((v_1,1,3),(v_1,1,3))\} \\
    H^M_{D[v_k,k], D[v_k,k]} &= \{((v_k,k,1),(v_k,k,1))\}
                        \cup \{((v_k,k,c),(v_k,k,c')) \mid c,c'\in\{2,3\}\}\,,
\end{align*}
and, therefore,
\begin{align*}
    R = \{((v_1,1,c),(v_k,k,c')) \mid c,c'\in[3]\}
         \setminus \{((v_1,1,3),(v_k,k,1))\}\,,
\end{align*}
which is not rectangular, as required.

For the reverse direction of the proof, suppose that 
there is an $\subclo{\lists}$-$M$-derectangularising sequence
$D_1,\ldots,D_m$. The fact that the sequence is derectangularising  
implies that $|D_i|\geq 2$ for each $i\in[m]$ --- 
see the remarks following Definition~\ref{def:derect}.
Each set in the sequence is a subset of some $D[v,j]$ in $\lists$
so for every $i\in[m]$ 
let $v_i$ denote the vertex in $[n]$
and let $j_i$ denote the index in $[k]$
such that  
$D_i \subseteq D[v_i,j_i]$.
Clearly, it is possible 
to have $(v_i,j_i)=(v_{i'},j_{i'})$ for distinct $i$ and
$i'$ in~$[m]$.

We will finish the proof by showing that $G$ has a size-$k$ independent set.
Let
\begin{equation*}
     R = H^M_{D_1 ,  D_2} \circ \dots \circ H^M_{D_{m-1}, D_m  },
\end{equation*}
which is not rectangular because the sequence is $\subclo{\lists}$-$M$-derectangularising.
Since $\{D_1,\ldots,D_m\}$  is
$M$-purifying, 
and any submatrix of~$\Mid$ with at least two rows and at least two columns is impure,
every pair $(i,i')\in [m]^2$  satisfies 
$M|_{D[v_i,j_i]\times D[v_{i'},j_{i'}]} \neq \Mid$.  
This means that we cannot have $(v_i,v_{i'})\in E(G)$
for any pair $(i,i')\in [m]^2$ 
 so the set $I=\{v_1, \dots, v_m\}$ is independent
in~$G$.  It remains to show that $|I|\geq k$.

Observe that, if $v_i=v_{i'}$, we must have $j_i=j_{i'}$ since, otherwise,
the construction ensures that
$$M|_{D[v_i,j_i]\times D[v_{i'},j_{i'}]}
 = M|_{D[v_i,j_i]\times D[v_i,j_{i'}]}
  = \Mid,$$ which we already ruled out.   
Therefore, $|I| \geq |\{j_1, \dots, j_m\}|$.

We must have $|j_i-j_{i+1}|\leq 1$ for each $i\in[m-1]$ as, otherwise,
$M|_{D[v_i,j_i]\times D[v_{i+1},j_{i+1}]}
 = \Mzero$, which implies that
$R=\emptyset$, which is rectangular.  There must be at least one $i\in
[m-1]$ such that $v_i=v_{i+1}$ and $j_i=j_{i+1}=1$, so
$M|_{D[v_i,j_i]\times D[v_{i+1},j_{i+1}]}
 = \Mstart$.  If not, $R$ is a composition
of relations corresponding to $\Mbij$ and $\Mend$ and any such
relation is either 
a bijection, or of the form of $\Mend$,
so it is rectangular.
Similarly, there must be at least one~$i$ such that $v_i=v_{i+1}$ and
$j_i=j_{i+1}=k$, giving $M|_{D[v_i,j_i]\times D[v_{i+1},j_{i+1}]}
 = \Mend$.
Therefore, the sequence $j_1, \dots, j_m$ contains 1 and~$k$.  Since
$|j_i-j_{i+1}|\leq 1$ for all $i\in[m-1]$, it follows that 
$[k] \subseteq \{j_1,\dots, j_m\}$, 
so $|I|\geq k$, as required.
In fact, $\{j_1,\dots, j_m\} = [k]$ since each $j_i\in [k]$ by construction. \end{proof}

We defined the problem \ExistsDerectSeq{} using a concise input
representation: $\subclo{\lists}$ does not need to be written out in full.
Instead, the instance is a subset~$\lists$ containing the maximal elements of~$\subclo{\lists}$.
For example, when the instance is~$\lists=\{D\}$,
we have $\subclo{\lists}= \powerset{D}$.
It is important to note that the \NP{}-completeness
of \ExistsDerectSeq{} is not an artifact of this concise input coding.
The elements of the list~$\lists$ constructed in the NP-hardness proof 
have length at most three, so 
the list $\subclo{\lists}$ could also be constructed explicitly in polynomial time.

Lemma~\ref{lem:small_derect} has the following immediate
corollary for the complexity of the dichotomy criterion of the general
\nListPartitions{M} problem.  Recall that, in this version of the
meta-problem, the input is just the matrix~$M$.

\begin{repcorollary}{cor:meta}\metacor\end{repcorollary}
\begin{proof} 
    Take $\lists= \{D\}$ in Lemma~\ref{lem:small_derect}.
\end{proof}

\section{Cardinality constraints}
\label{sec:card}

Finally, we show how lists can be used to implement cardinality
constraints of the kind that often appear in counting problems in
combinatorics.

Feder, Hell, Klein and Motwani~\cite{FHKM} point out that lists can be
used to determine whether there are $M$-partitions that obey simple
cardinality constraints.  For example, it is natural to require some
or all of the parts to be non-empty or, more generally, to contain at
least some constant number of vertices.  Given a $D\times D$
matrix~$M$, we represent such cardinality constraints as a function
$C\colon D\to\nats$.  We say that an $M$-partition $\sigma$ of a
graph~$G$ \emph{satisfies} the constraint if, for each $d\in D$, $|\{v\in
V(G)\mid \sigma(v)=d\}| \geq C(d)$.  Given a cardinality constraint
$C$, we write $|C| = \sum_{d\in D} C(d)$. 

We can determine whether there is an $M$-partition of~$G=(V,E)$ that
satisfies the cardinality constraint~$C$ by making at most
${|V|}^{|C|}$
queries to an oracle for the list $M$-partitions
problem, as follows.  Let $L_C$ be the set of list functions $L\colon
V\to \powerset{D}$ such that:
\begin{itemize}
\item for all $v\in V\!$, either $L(v) = D$ or $|L(v)| = 1$, and
\item for all $d\in D$, there are exactly $C(d)$ vertices~$v$ with
    $L(v) = \{d\}$.
\end{itemize}
There are 
at most ${|V|}^{|C|}$
such list functions and it is clear that
$G$~has an $M$-partition satisfying~$C$ if, and only if, it has a list
$M$-partition that respects at least one $L\in L_C$.  
The number of queries 
is polynomial in $|V|$ as long as the cardinality
constraint~$C$ is independent of~$G$.

For counting, the situation is a little more complicated, as we must
avoid double-counting.  The solution is to count all $M$-partitions
of the input graph and subtract off those that fail to satisfy the
cardinality constraint.  We formally define the problem
\CMPartitions{C}{M} as follows, parameterized by a $D\times D$
matrix~$M$ and a cardinality constraint function~$C\colon D\to \nats$.

\prob{\CMPartitions{C}{M}.}
  {A graph $G$.}
  {The number of $M$-partitions of~$G$ that satisfy~$C$.}

\begin{proposition}
\label{prop:cardinality}
    \CMPartitions{C}{M} is polynomial-time Turing reducible to
    \nListPartitions{M}.
\end{proposition}
\begin{proof}
    Given the cardinality constraint function~$C$, let $R = \{d\in
    D\mid C(d)>0\}$: that is, $R$ is the set of parts that have a
    non-trivial cardinality constraint.  For any set $P\subseteq R$,
    say that an $M$-partition~$\sigma$ of a graph~$G=(V,E)$
    \emph{fails on~$P$} if $|\{v\in V\mid \sigma(v) = d\}| < C(d)$ for
    all~$d\in P$. That is, if $\sigma$~violates the cardinality
    constraints on all parts in~$P$ (and possibly others, too).
      Let $\Sigma$ be the set of all $M$-partitions of our given input
    graph~$G$.  
For $i\in R$, let $A_i = \{\sigma \in \Sigma \mid \mbox{$\sigma$ fails on $\{i\}$}\}$ and
let $A=\bigcup_{i\in R} A_i$.
By inclusion-exclusion,
\begin{align*}
|A| &= -\!\!\sum_{\emptyset \subset P \subseteq R} {(-1)}^{|P|} \left|\bigcap_{i\in P} A_i\right|\\
&= -\!\!\sum_{\emptyset\subset P \subseteq R} {(-1)}^{|P|} \big|\{\sigma \in \Sigma \mid
\mbox{$\sigma$ fails on $P$}\}\big|\,.
\end{align*}

  We wish to compute
    \begin{align*}
        \big|\{\sigma\in \Sigma\mid \text{$\sigma$ satisfies $C$}\}\big|
            &= \big|\Sigma\big| - |A| \\
            &= \big|\Sigma\big|
               + \sum_{\emptyset \subset P\subseteq R}
                     (-1)^{|P|}
                     \big|\{\sigma\in\Sigma \mid \text{$\sigma$
                           fails on $P$}\}\big|\,.
    \end{align*}

    Therefore, it suffices to show that we can use lists to count the
    $M$-partitions that fail on each non-empty $P\subseteq R$.  For
    such a set~$P$, let $L_P$ be the set of list functions~$L$ such
    that
    \begin{itemize}
    \item for all $v\in V$, either $L(v) = D\setminus P$ or $L(v) =
        \{p\}$ for some $p\in P$, and
    \item for all $p\in P$, $\big|\big\{v\in V\mid L(v)=\{p\}\big\}\big| < C(p)$.
    \end{itemize}
    Thus, the set of 
$M$-partitions that respect some $L\in L_P$
    is precisely the set of 
$M$-partitions
that fail on~$P$.
Also, for distinct $L$ and~$L'$ in $L_P$, the set of $M$-partitions that respect~$L$ is disjoint
from the set of $M$-partitions that respect~$L'\!$. So we can compute 
$ \big|\{\sigma\in\Sigma \mid \text{$\sigma$
                           fails on $P$}\}\big|$
by making $|L_P|$ calls to 
\nListPartitions{M}, noting that $|L_P|\leq |V|^{|C|}\!$.
\end{proof}

As an example of a combinatorial structure that can be represented as
an $M$-partition problem with cardinality constraints, consider the
\emph{homogeneous pairs} introduced by Chv\'atal and
Sbihi~\cite{CS1987:Bull-free}.  A homogeneous pair in a graph
$G=(V,E)$ is a partition of $V$ into sets $U$, $W_1$ and $W_2$ such
that:
\begin{itemize}
\item $|U|\geq 2$;
\item $|W_1|\geq 2$ or $|W_2|\geq 2$ (or both);
\item for every vertex~$v\in U$, $v$~is either adjacent to every
    vertex in $W_1$ or to none of them; and
\item for every vertex~$v\in U$, $v$~is either adjacent to every
    vertex in $W_2$ or to none of them.
\end{itemize}

Feder et al.~\cite{FHKM} observe that
the problem 
of determining whether a graph has a homogeneous pair
can be represented as 
the problem of determining whether it has an
\Mhp{}-partition satisfying certain constraints,
where $D = \{1, \dots, 6\}$ and
\begin{equation*}
    \Mhp = \begin{pmatrix}
               * & * & 1 & 0 & 1 & 0 \\
               * & * & 1 & 1 & 0 & 0 \\
               1 & 1 & * & * & * & * \\
               0 & 1 & * & * & * & * \\
               1 & 0 & * & * & * & * \\
               0 & 0 & * & * & * & *
           \end{pmatrix}.
\end{equation*}  

$W_1$ corresponds to the set of vertices mapped to part~$1$ (row 1 of $\Mhp$),
$W_2$ corresponds to the set of vertices mapped to part~$2$ (row 2 of $\Mhp$),
and $U$ corresponds to the set of vertices mapped to parts~$3$--$6$. 

In fact, there is a one-to-one correspondence between the homogeneous pairs of~$G$
in which $W_1$ and~$W_2$ are non-empty and the $\Mhp$-partitions $\sigma$ of~$G$
that satisfy the following additional constraints.  For $d\in D$, let $N_\sigma(d) = |\{v\in V(G)\mid \sigma(v)=d\}|$ be the number of vertices that $\sigma$~maps to part~$d$.  We require that
\begin{itemize}
\item $N_\sigma(3) + N_\sigma(4) + N_\sigma(5) + N_\sigma(6)\geq 2$,
\item $N_\sigma(1) > 0$ and $N_\sigma(2) > 0$, and
\item at least one $N_\sigma(1)$ and $N_\sigma(2)$ is at least~$2$.
\end{itemize}
To see this, consider a homogeneous pair $(U,W_1,W_2)$
in which~$W_1$ and~$W_2$ are non-empty. Note that there is exactly one $\Mhp$-partition
of~$G$ in which vertices in~$W_1$ are mapped to part~$1$ and vertices in~$W_2$
are mapped to part~$2$ and vertices in~$U$ are mapped to parts~$3$--$6$.
There is exactly one part available to each $v\in U$ since $v$ has an edge or non-edge
to~$W_1$ (but not both!) ruling out exactly two parts
and $v$ has an edge or non-edge to~$W_2$ ruling out an additional part.
Going the other way, an $\Mhp$-partition that satisfies the constraints includes a homogeneous pair.

Now let 
\begin{equation*}
    \Mhs = \begin{pmatrix}
    * & 0 & 1\\
    0 & * & *\\
    1 & * & * 
           \end{pmatrix}.
\end{equation*}  
There is a one-to-one correspondence 
between the homogeneous pairs of~$G$
in which  $W_2$ is empty and the $\Mhs$-partitions of~$G$
that satisfy the following additional constraints.
\begin{itemize}
\item At least two vertices are mapped to parts~$2$--$3$ (vertices in these parts are in~$U$).
\item At least two vertices are mapped to part~$1$ (vertices in this part are in~$W_1$).
\end{itemize} 
Symmetrically, there is also a one-to-one correspondence 
between the homogeneous pairs of~$G$
in which  $W_1$ is empty and the $\Mhs$-partitions of~$G$
that satisfy the  above constraints.  (Partitions according to $\Mhs$
correspond to so-called ``homogeneous sets'' but we do not need the details of
these.)

It is known from \cite{EKR1997:Hom-pair} that, in deterministic
polynomial time, 
it is possible to determine whether a graph
contains a homogeneous pair and, if so, to find one.  We show that the
homogeneous pairs in a graph can also be counted 
in polynomial time.
We start by considering the relevant list-partition counting problems.
 
\begin{theorem}
\label{thm:hompair}
    There  are polynomial-time algorithms for \nListPartitions{\Mhp} and \nListPartitions{\Mhs}.
\end{theorem}
\begin{proof}
We first show that there is a polynomial-time algorithm for \nListPartitions{\Mhp}.
The most natural way to do this
would be to show that there is no $\powerset{D}$-\Mhp{}-derectangularising sequence
and then apply Theorem~\ref{thm:explicitdichotomy}.
In theory, we could show that there is no $\powerset{D}$-\Mhp{}-derectangularising sequence
  by brute force since $|D|=6$, but
the number of possibilities is too large to make this feasible. Instead, we argue non-constructively.

First,
    if there is no $\powerset{D}$-\Mhp{}-derectangularising sequence,
    the result follows from Theorem~\ref{thm:explicitdichotomy}.

    Conversely, suppose that $D_1, \dots, D_k$ is a
    $\powerset{D}$-\Mhp{}-derectangularising sequence.  Let $M$~be the
    matrix such that $M_{i,j} = 0$ if $(\Mhp)_{i,j} = 1$ and $M_{i,j}
    = (\Mhp)_{i,j}$, otherwise.  $D_1, \dots, D_k$ is also a
    $\powerset{D}$-$M$-derectangularising sequence, since $H^M_{X,Y} =
    H^{\Mhp}_{X,Y}$ for any $X,Y\subseteq D$ and any sequence $D_1,
    \dots, D_k$ is $M$-purifying because $M$~is already pure.
    Therefore, by Theorem~\ref{thm:explicitdichotomy}, counting list
    $M$-partitions is \numP{}-complete.

    However, counting the list $M$-partitions of a graph~$G$
    corresponds to counting list homomorphisms from~$G$ to the $6$-vertex
    graph~$H$ whose two components are an edge and a $4$-clique, and
    which has loops on all six vertices.  There is a very
    straightforward polynomial-time algorithm for this problem (a
    simple modification of the version without lists in~\cite{DG}).
    Thus, $\numP=\FP$ so, in particular, there is a polynomial-time
    algorithm for counting list \Mhp{}-partitions.

    The proof that there is a polynomial-time algorithm for
    \nListPartitions{\Mhs} is similar.
\end{proof}

\begin{corollary}
\label{cor:hompair}
    There is a polynomial-time algorithm for counting the homogeneous
    pairs in a graph.
\end{corollary}
\begin{proof}
    We are given a graph $G=(V,E)$ and we wish to compute the number
    of homogeneous pairs that it contains.
By the one-to-one correspondence given earlier, it suffices to show how to count
$\Mhp$-partitions and $\Mhs$-partitions of~$G$ satisfying additional constraints.
We start with the first of these.
Recall the constraints on the $\Mhp$-partitions $\sigma$ that we wish to count:

\begin{itemize}
\item $N_\sigma(3) + N_\sigma(4) + N_\sigma(5) + N_\sigma(6)\geq 2$,
\item $N_\sigma(1) > 0$ and $N_\sigma(2) > 0$, and
\item at least one $N_\sigma(1)$ and $N_\sigma(2)$ is at least~$2$.
\end{itemize}

    Define three subsets $\Sigma_1$, $\Sigma_2$ and $\Sigma_{1,2}$
    of the set of 
\Mhp{}-partitions of~$G$ that satisfy the constraints.  
In the definition of each of $\Sigma_1$, $\Sigma_2$ and 
$\Sigma_{1,2}$,
we will require that
parts~$1$ and~$2$ are non-empty and
    parts~$3$--$6$ contain a total of at least two
    vertices.  In~$\Sigma_1$, part~$1$ must contain at least two
    vertices; in~$\Sigma_2$, part~$2$ must contain at least two
    vertices; in~$\Sigma_{1,2}$, both parts $1$ and~$2$ must contain
    at least two vertices.  The number of 
suitable $\Mhp$-partitions of~$G$
is
    $|\Sigma_1| + |\Sigma_2| - |\Sigma_{1,2}|$.

    Each of $|\Sigma_1|$, $|\Sigma_2|$ and $|\Sigma_{1,2}|$ can be
    computed by counting the $\Mhp$-partitions of~$G$ that satisfy
    appropriate cardinality constraints.  Parts $1$ and~$2$ are
    trivially dealt with.  The requirement that parts $3$--$6$ must
    contain at least two vertices between them is equivalent to saying
    that at least one of them must contain at least two vertices or at
    least two must contain at least one vertex.  This can be expressed
    with a sequence of cardinality constraint functions and using
    inclusion--exclusion to eliminate double-counting.

Counting  constrained $\Mhs$-partitions of~$G$ is similar (but simpler).
\end{proof}

\bibliographystyle{plain}
\bibliography{\jobname}{} 

\end{document}